\documentclass[a4paper,UKenglish,cleveref,thm-restate]{lipics-v2021}

\usepackage[subrefformat=simple,labelformat=simple,position=b]{subcaption}

\usepackage{nicefrac}
\usepackage{mathtools}
\usepackage{xspace}
\usepackage{complexity}

\usepackage{graphicx}

\usepackage{amsmath}
\usepackage{amsthm}

\usepackage{amssymb}
\usepackage{mathtools}
\usepackage{amsfonts}
\usepackage{siunitx}
\usepackage{nicefrac}

\newcommand{\polygon}[0]{\mathcal{P}}
\newcommand{\guardset}[0]{\mathcal{G}}
\newcommand{\psat}[0]{\textsc{P3Sat}$_{\overline{\underline{3}}}$\xspace}
\newcommand{\agp}[0]{AGP\xspace}
\newcommand{\dagp}[0]{\textsc{Dispersive AGP}\xspace}

\crefname{figure}{Figure}{Figures}
\crefname{theorem}{Theorem}{Theorems}
\crefname{lemma}{Lemma}{Lemmas}
\crefname{corollary}{Corollary}{Corollaries}
\crefname{section}{Section}{Sections}
\crefname{appendix}{Appendix}{Appendices}
\crefname{remark}{Remark}{Remarks}
\crefname{claim}{Claim}{Claims}
\crefname{conjecture}{Conjecture}{Conjectures}
\crefname{observation}{Observation}{Observations}

\graphicspath{{./figures/}}

\title{Guarding Offices with Maximum Dispersion}
\titlerunning{Guarding Offices with Maximum Dispersion}

\author{Sándor P. Fekete}{Department of Computer Science, TU Braunschweig, Germany}{s.fekete@tu-bs.de}{https://orcid.org/0000-0002-9062-4241}{}
\author{Kai Kobbe}{Department of Computer Science, TU Braunschweig, Germany}{kobbe@ibr.cs.tu-bs.de}{https://orcid.org/0009-0008-1772-1409}{}
\author{Dominik Krupke}{Department of Computer Science, TU Braunschweig, Germany}{krupke@ibr.cs.tu-bs.de}{https://orcid.org/0000-0003-1573-3496}{}
\author{Joseph S. B. Mitchell}{Department of Applied Mathematics and Statistics, Stony Brook University, USA}{joseph.mitchell@stonybrook.edu}{https://orcid.org/0000-0002-0152-2279}{}
\author{Christian Rieck}{Department of Discrete Mathematics, University of Kassel, Germany}{christian.rieck@mathematik.uni-kassel.de}{https://orcid.org/0000-0003-0846-5163}{}
\author{Christian Scheffer}{Department of Electrical Engineering and Computer Science, Bochum University of Applied~Sciences, Germany}{christian.scheffer@hs-bochum.de}{https://orcid.org/0000-0002-3471-2706}{}

\authorrunning{S. P. Fekete, K. Kobbe, D. Krupke, J. S. B. Mitchell, C. Rieck, and C. Scheffer}
\supplement{\url{https://github.com/KaiKobbe/dispersive_agp_solver}}
\Copyright{Sándor P.\ Fekete and Kai Kobbe and Dominik Krupke and Joseph S.\,B.\ Mitchell and Christian~Rieck and Christian Scheffer}



\funding{Work from TU Braunschweig was partially supported by the German Research Foundation~(DFG), project ``CG:SHOP'', FE 407/21-1. 
Joseph Mitchell was partially supported by the National Science Foundation (CCF-2007275). Christian Rieck was partially supported by DFG grant 522790373.}

\nolinenumbers
\hideLIPIcs

\begin{document}

    \maketitle

    \begin{abstract}
		We investigate the \textsc{Dispersive Art Gallery Problem} with vertex guards and rectangular visibility ($r$-visibility) for a class of orthogonal polygons that reflect the properties of real-world floor plans: these \emph{office-like} polygons consist of rectangular rooms and corridors.
		In the dispersive variant of the Art Gallery Problem, the objective is not to minimize the number of guards but to maximize the minimum geodesic $L_1$-distance between any two guards, called the \emph{dispersion distance}. 
		
		Our main contributions are as follows. 
		We prove that determining whether a vertex guard set can achieve a dispersion distance of $4$ in office-like polygons is \NP-complete, where vertices of the polygon are restricted to integer coordinates. 
		Additionally, we present a simple worst-case optimal algorithm that guarantees a dispersion distance of $3$ in polynomial time. 
		Our complexity result extends to polyominoes, resolving an open question posed by Rieck and Scheffer~\cite{rieck-scheffer-dispersiveAGP}. 
		When vertex coordinates are allowed to be rational, we establish analogous results, proving that achieving a dispersion distance of $2+\varepsilon$ is \NP-hard for any $\varepsilon > 0$, while the classic Art Gallery Problem remains solvable in polynomial time for this class of polygons. 
		Furthermore, we give a straightforward polynomial-time algorithm that computes worst-case optimal solutions with a dispersion distance~$2$.

		On the other hand, for the more restricted class of hole-free independent office-like polygons, we propose a dynamic programming approach that computes optimal solutions. 
		Moreover, we demonstrate that the problem is practically tractable for arbitrary orthogonal polygons. 
		To this end, we compare solvers based on SAT, CP, and MIP formulations. 
		Notably, SAT solvers efficiently compute optimal solutions for randomly generated instances with up to \num{1600} vertices in under \SI{15}{\second}. 
    \end{abstract}
	\newpage
\section{Introduction}

The \textsc{Art Gallery Problem} (AGP) is a widely studied problem in computational geometry. 
Given a polygonal region, the objective is to select a minimum number of points (guards) such that every point within the region is visible from at least one guard, with visibility defined by an unobstructed line segment that does not intersect the polygon boundary. 
Since its introduction, numerous variants of the problem have been studied~\cite{orourke,shermer,urrutia}.

A relevant application arises in the deployment of sensors across a designated area, for which the challenge is to balance sensor spacing and signal coverage. 
On one hand, ensuring sufficient distance between sensors helps mitigate interference, preserving data integrity. 
On the other hand, it is essential to guarantee full coverage, ensuring continuous and reliable signal transmission across the entire area.
This trade-off is captured by the \textsc{Dispersive Art Gallery Problem}, which is defined as follows: 
Given a polygon $\polygon$ and a rational number $\ell$, determine whether there exists a guard set $\guardset$ for $\polygon$ such that the geodesic distances between any two guards in $\guardset$ are at least $\ell$. 
Notably, this formulation focuses solely on pairwise distances between guards rather than minimizing their total number.

In this paper, we investigate the theoretical and practical tractability of this problem in the context of vertex guards within a class of orthogonal polygons that model key structural properties of real-world floor plans or rectangular galleries~\cite{cruz}, consisting of orthogonal
rooms and corridors, called~\emph{office-like polygons}\footnote{This class of polygons was originally introduced by Cruz and Tom\'as~\cite{cruz} using the Portuguese expression SCOT for \emph{salas e corredores ortogonais}, i.e., orthogonal rooms and corridors.}; see~\cref{fig:initial_example} for an example.

\begin{figure}[htb]
    \centering
    \includegraphics[page=2]{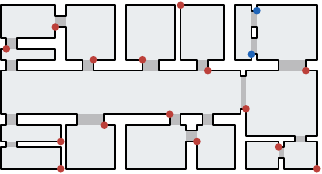}
    \caption{An office-like polygon with holes, along with a vertex guard set (shown as red and blue points) achieving maximum dispersion. Rooms are highlighted in light gray, and corridors in dark gray. The pair of blue guards represents the minimum geodesic distance among all pairs in the set.}
    \label{fig:initial_example}
\end{figure}

\subsection{Our contributions}
We present results for the \dagp with vertex guards in orthogonal polygons, including general, polyominoes, and office-like polygons, under \mbox{$r$-visibility} and $L_1$-geodesics. 
Without loss of generality, we assume a minimum distance of $1$ between any pair of vertices.\footnote{This assumption serves purely to simplify the statements; alternatively, setting the minimum distance to $\zeta$ would generalize all results proportionally.}

\begin{itemize}
    \item 
    {We show that the ratio between the cardinality of a maximally dispersed guard set and the smallest possible guard set can be arbitrarily large, even in histograms.
    }
    \item 
    {We present polynomial-time algorithms that compute worst-case optimal guard sets for office-like polygons. Specifically, the algorithms achieve dispersion distances of~$3$ when vertices are restricted to integer coordinates, and~$2$ otherwise. These bounds are tight, as demonstrated by corresponding lower-bound constructions.
    }
    \item 
    {Deciding whether a given office-like polygon allows a guard set with a dispersion distance of $4$ is \NP-complete when vertices are restricted to integer coordinates; 
    the same result holds for polyominoes (and with minor modifications even for thin polyominoes), thereby resolving an open question from \cite{rieck-scheffer-dispersiveAGP} on \NP-hardness in that setting.
    If vertex coordinates are allowed to be rational, deciding whether a dispersion distance of at least~$2+\varepsilon$ can be achieved for any~$\varepsilon > 0$ is already hard; see~\cref{thm:hardness-integer,cor:hardness-polyominoes,cor:hardness-arbitrary}.}
    \item 
    {For independent hole-free office-like polygons, we provide a dynamic programming approach to compute guard sets with maximum dispersion in polynomial time;~see~\cref{thm:dpalgo}.}
    \item 
    {We evaluate practical tractability by utilizing SAT, CP, and MIP techniques.
    Our~experiments show that instances with up to \num{1600} vertices can be solved in under \SI{15}{\second}.}
\end{itemize}

Note that the assumption on the minimum distance between any pair of vertices is not needed in both the practical evaluation as well as the dynamic programming. 
However, we adopt this assumption throughout the paper to improve clarity and readability.
Details for statements marked with ($\star$) are deferred to the appendix.

\subsection{Previous work}
The \textsc{Art Gallery Problem} (AGP) has been a cornerstone of computational geometry, inspiring extensive research. 
For an overview, see the previously mentioned book by O'Rourke~\cite{orourke} or the surveys by Shermer~\cite{shermer} and Urrutia~\cite{urrutia}. 

For a long time, it was known that the classic AGP is \NP-hard in several fundamental variants~\cite{LeeL86,SchuchardtH95}. 
Abrahamsen, Adamaszek, and Miltzow~\cite{irrational-guards} were the first to show that, even in monotone polygons, ``irrational guards are sometimes needed''. 
They subsequently established \mbox{$\exists \mathbb{R}$-completeness}~\cite{abrahamsenER}. 
Meijer and Miltzow~\cite{meijer2024irrationalguardsneeded} construct a polygon that can be guarded by two guards, both of which must have irrational coordinates.
When considering $r$-visibility (where two points $u$ and $v$ see each other if and only if the rectangle described by $u$ and $v$ lies completely within the polygonal region) instead of the classic line visibility, the AGP remains \NP-hard for polygons with holes~\cite{iwamoto}, whereas optimal guard sets can be computed in polynomial time for hole-free (orthogonal) polygons~\cite{worman}. 
Recently, Cruz and Tom{\'{a}}s~\cite{cruz} introduced a family of orthogonal polygons designed to mimic properties of real-world buildings. 
We will refer to these polygons as \emph{office-like polygons}. 
These polygons consist of multiple rectangular rooms connected by rectangular corridors, with each corridor being narrower than both of its incident rooms. 
Using \mbox{$r$-visibility}, they obtained several results. 
They introduced a simple greedy algorithm that computes optimal guard placements using the minimum number of guards for simple office-like polygonal environments.
They also proved that efficient algorithms exist for the subclass of \emph{$r$-independent} office-like polygons (i.e., office-like polygons in which the extensions of all pairs of adjacent corridors with the same direction are either coincident or disjoint), utilizing matching and flow techniques. 
However, guarding general office-like polygons under $r$-visibility remains \NP-hard.

All of the aforementioned results focus on minimizing the number of guards required for coverage. However, some variants relax this constraint and shift the focus to alternative objectives or other factors.
One such example is the \textsc{Chromatic AGP}~\cite{erickson2010chromatic,EricksonL11,FeketeFHM014,IwamotoI20}, with each guard being assigned a color. 
The goal is to minimize the number of colors used so that no two guards of the same color class have overlapping visibility regions while ensuring full coverage of the domain. 
In the \textsc{Conflict-free Chromatic AGP}~\cite{BartschiGMTW14,BartschiS14,hksvw-ccgoag-18}, the aforementioned overlapping constraint is relaxed, requiring that every point in the domain must see at least one guard whose color is unique among all guards visible from that point.

Under the current definition of the dispersive variant, the cardinality of a guard set is not considered relevant either.
Rieck and Scheffer~\cite{rieck-scheffer-dispersiveAGP} studied the \dagp for vertex guards in polyominoes, assuming $L_1$ geodesics and $r$-visibility. 
They proved that deciding whether a given thin polyomino allows a guard set with a dispersion distance of~\num{5} is \NP-complete. 
Additionally, they provided a linear-time algorithm that produces a worst-case optimal solution for simple polyominoes, ensuring a guard set with a dispersion distance of~\num{3}. 
Fekete et al.~\cite{dispersiveAGP-CCCG24} extended this to vertex guards in polygons. 
They demonstrated that determining whether a guard set with a dispersion distance of~\num{2} exists in polygons with holes is \mbox{\NP-complete}. 
Furthermore, they showed that a dispersion distance of~\num{2} is sometimes unavoidable even in hole-free polygons and proposed an algorithm to generate such sets. 

Dispersion problems, which are closely related to packing problems, focus on arranging a set of objects to be as far apart as possible or selecting a subset whose elements are maximally separated.
Such problems naturally arise in the context of the obnoxious facility location problem~\cite{cappanera1999survey} and the problem of distant representatives~\cite{FialaKP05}. 
Baur and Fekete~\cite{baur} studied the problem of distributing a set of points within a polygonal domain to maximize dispersion. 
They~established that this problem is \NP-hard and ruled out the possibility of a \PTAS\ for geometric dispersion problems unless \P=\NP. 

\subsection{Preliminaries and first observations}\label{subsec:preliminaries}
Given any polygon~$\polygon$, a subset $\guardset \subset \polygon$ of points is called a \emph{guard set} for $\polygon$ if every point~$p \in \polygon$ is \emph{visible} by at least one point~$g \in \guardset$. 
We restrict ourselves to vertex guards, i.e., guards are only placed on vertices of the domain.
We adopt the \emph{r-visibility} model, where two points $p$ and $q$ are said to \emph{see each other} if the axis-aligned rectangle defined by $p$ and $q$ is entirely contained within~$\polygon$.
The \emph{visibility region}~$\emph{Vis}(p) \subseteq \polygon$ of a point~$p$ is the set of all points that are visible from~$p$.
The \emph{size} of a guard set is its number of guards.
We measure the distance~$\delta(p,q)$ between two points~$p,q \in \polygon$ as \emph{$L_1$-geodesics}.
With this, the \emph{dispersion distance} $\ell$ of a guard set~$\guardset$ is the minimum  geodesic $L_1$-distance between any pair of guards; we say that a~guard set \emph{realizes} a dispersion distance of $\ell$, or is \emph{$\ell$-realizing}.

In this paper, we consider office-like polygons as well as polyominoes.
\emph{Polyominoes} are orthogonal polygons formed by joining unit squares edge to edge.
An \emph{office-like polygon}~${\polygon = (\mathcal R, \mathcal C)}$ is a connected orthogonal polygon made up from a collection of rectangular rooms $\mathcal{R}$ that are linked by a collection of rectangular corridors $\mathcal{C}$. 
A~corridor connects two rooms if a side of the corridor is \emph{strictly contained} in an edge of both rooms, i.e., the corridor is \emph{narrower} than the rooms, and none of the vertices of a room coincides with a vertex of a corridor. 
An example of an office-like polygon with several rooms and corridors, as well as a set of guards is depicted in~\cref{fig:initial_example,fig:3_sufficient_main}. 
Each corridor connects exactly two rooms. 
A corridor is \emph{horizontal} if it links rooms to its left and right; otherwise, it is \emph{vertical}.
Two corridors $C_1$ and $C_2$ are \emph{independent} if no guard placed on a vertex of $C_1$ also guards $C_2$.
If~every pair of corridors is independent, then the polygon is~\emph{independent}.

Fekete et al.~\cite{dispersiveAGP-CCCG24} showed that optimal solutions for the \dagp when considering arbitrary simple polygonal domains may contain many guards, even if the domain can be covered by a small number of guards.
In particular, they obtained a family of polygons in which the ratio between the sizes of the respective guard sets is unbounded.
We strengthen this by observing that this is true even for histograms, as visualized in~\cref{unbounded_dagp}.

\begin{figure}[htb]
	\centering
	\begin{subfigure}[b]{0.45\textwidth}
		\centering
		\includegraphics[page=1]{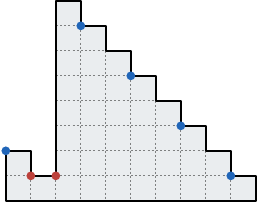}
		\subcaption{}
		\label{fig:dagp_agp_unbounded}
	\end{subfigure}
	\hfill
	\begin{subfigure}[b]{0.45\textwidth}
		\centering
		\includegraphics[page=2]{figures/dagp_agp_unbounded.pdf}
		\subcaption{}
		\label{fig:extension_unbounded}
	\end{subfigure}
	\caption{A family of histogram polygons illustrating that the ratio between optimal solutions of the AGP and the dispersive AGP can differ significantly. Red and blue points indicate optimal vertex guard placements for AGP and dispersive AGP, respectively.}
	\label{unbounded_dagp}
\end{figure}

In the case of office-like polygons, it appears that the situation may differ.
However, we can construct a family of these polygons for which this ratio asymptotically approaches~$2$; the construction is given in~\cref{app:small-guard-sets}.

\section{Worst-case optimality} \label{chap:wc_opt}

In this section, we discuss worst-case optimal algorithms for the \dagp in office-like polygons. 
We begin by examining them with vertices at integer coordinates.
Afterward, we extend our analysis to present similar results for the unrestricted case.

\subsection{Office-like polygons with vertices at integer coordinates}\label{vpic_main}
We establish the following worst-case bound.

\begin{restatable}[$\star$]{theorem}{worstCaseIntegerScots}
	\label{thm:worst-case-integer-scots}
	For office-like polygons with holes and vertices at integer coordinates, there always exists a guard set with dispersion distance at least~$3$, which is sometimes optimal.
\end{restatable}

As a first observation, it is easy to construct instances in which every guard set has a dispersion distance of at most~$3$; an example is visualized in~\cref{fig:3_necessary_main}. 
In this polygon, at least one guard is required in each of the three unit-square corridors, but then at least one pair of guards must be in distance $3$ of one another.
Notably, this polygon can serve as some kind of building block, as it can be extended arbitrarily along the black arrows.

We now outline the high-level idea of a polynomial-time algorithm that computes guard sets with a dispersion distance of $3$, thereby ensuring optimal solutions in the worst-case.
Our approach consists of three phases. 
In each phase, we place a set of guards, guaranteeing that no guard is positioned less than $3$ units away from any previously placed guard. 

\begin{description}
	\item[Phase (1):] {Place guards on vertices of corridors connecting two rooms \emph{vertically} to ensure that these corridors are properly guarded.}
	\item[Phase (2):] {Place guards similarly on the vertices of corridors connecting rooms \emph{horizontally}.}
	\item[Phase (3):] {Place guards in any remaining rooms that are not yet covered.}
\end{description}

An example illustrating the guard placements made by our algorithm is shown in~\cref{fig:3_sufficient_main}. 
Red guards are placed during the first phase, blue guards during the second phase, and green guards during the final phase.

\begin{figure}[htb]
	\centering
	\begin{subfigure}[b]{0.3\textwidth}
		\centering
		\includegraphics[page=1]{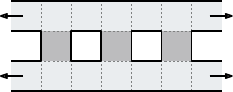}
		\subcaption{}
		\label{fig:3_necessary_main}
	\end{subfigure}
	\hfill
	\begin{subfigure}[b]{0.5\textwidth}
		\centering
		\includegraphics[page = 3]{figures/example_wcopt_main.pdf}
		\subcaption{}
		\label{fig:3_sufficient_main}
	\end{subfigure}
	\caption{A dispersion distance of~$3$ is (a) sometimes best possible and (b) always realizable.}
	\label{fig:wc_opt_3_main}
\end{figure}

In each phase, our approach ensures that every new guard is placed at least 3 units away from all previously placed guards.
Essentially, in the first two phases, we place guards either on the left side of vertical corridors or on the bottom side of horizontal corridors.
Moreover, it is easy to see that in every room that is not covered by the first two phases, there is at least one vertex that is in distance at least $3$ to a guard placed in an incident corridor.

We give all technical details and a proof of~\cref{thm:worst-case-integer-scots} in~\cref{app:worst-case-integer}.

\subsection{Office-like polygons with vertices at arbitrary rational coordinates}\label{lbpd_main}
We now investigate worst-case optimal solutions in case that vertices of the polygon are not restricted to integer coordinates. 
The primary challenge in achieving large dispersion distances arises in corridors that are independent but have a minimal horizontal or vertical separation. 
We refer to such arrangements of corridors as \emph{densely packed}; an example is depicted in~\cref{fig:2eps_necessary_main}.

\begin{figure}[htb]
	\centering
	\begin{subfigure}[b]{0.6\textwidth}
		\centering
		\includegraphics[page=2]{figures/wc_opt_examples_main.pdf}
		\subcaption{}
		\label{fig:2eps_necessary_main}
	\end{subfigure}
	\hfill
	\begin{subfigure}[b]{0.3\textwidth}
		\centering
		\includegraphics[page = 1]{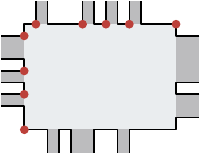}
		\subcaption{}
		\label{fig:2_sufficient_main}
	\end{subfigure}
	\caption{A dispersion distance of (a) $2 + \varepsilon$ for every rational $\varepsilon > 0$ is sometimes best possible, but (b) a dispersion distance of $2$ is always realizable.}
	\label{fig:wc_opt_2eps_main}
\end{figure}

Our main theorem for these polygons is the following.

\begin{theorem}
	\label{thm:worst-case-scots}
	For office-like polygons with holes, there always exists a guard set with a dispersion distance of at least~$2$. 
	Furthermore, for every rational $\varepsilon >0$, there is such a polygon in which a dispersion distance of $2+\varepsilon$ is best possible.
\end{theorem}

We begin with observing that arrangements of densely packed corridors enforce small dispersion distances.

\begin{restatable}[$\star$]{lemma}{disptwoeps}
    \label{lem:disp_2eps_example}
    For every $n \in \mathbb{N}$ and $\varepsilon > 0$, there exists an office-like polygon with at least~$n$ vertices such that every guard set has a dispersion distance that is smaller than $2+\varepsilon$.
\end{restatable}

Consider the polygon shown in~\cref{fig:2eps_necessary_main}. 
All corridors (depicted in dark gray) are unit squares, the rooms (represented as horizontal bars in light gray) have a height of $1$, and the red rectangles have a width of $0 < \tau < \nicefrac{\varepsilon}{2}$ for any $\varepsilon \in (0,1)$.

We now give the idea that any guard set contains a guard in corridor~$C$ only if its dispersion distance is smaller than~$2+\varepsilon$.
Clearly, every guard set must contain such a guard, as otherwise the corridor is not guarded.
Every guard set with a dispersion distance of at least~$2+\varepsilon$ that contains a guard at the bottom-left (top-left) vertex of~$C$ must contain all guards depicted in green (blue).
However, this leads to a contradiction, as no guard can be placed with sufficient distance in~$C_1$ (or~$C_2$), i.e., there is no $2+\varepsilon$-realizing guard set.
Incoming arcs indicate that a previously placed guard enforces a guard position (because no other option in the respective corridor is possible).
Due to symmetry, a guard set with adequate dispersion distance cannot contain a guard at either the bottom-right or top-right vertex of~$C$.
Moreover, this instance can easily be extended along the black arrows, to obtain office-like polygons with an arbitrary number of vertices. We give full details in~\cref{app:worst-case-integer}.

We now prove that a dispersion distance of~$2$ is always realizable.

\begin{lemma}\label{wc_proof_2}
    Every office-like polygon allows a guard set with dispersion distance at least~$2$.
\end{lemma}

\begin{proof}
    Our construction of a guard set with a dispersion distance of at least $2$ is as follows.
    In every room, we walk in clockwise direction along the boundary from the bottom-left to the top-right vertex, and place a guard at every other vertex, see~\cref{fig:2_sufficient_main}.

    This routine clearly yields a guard set. 
    Moreover, within each room guards are placed in geodesic distance at least $2$.
    We ensure that guards placed in different but adjacent rooms have at least a geodesic distance of $2$ by only placing guards at the left and top part of each room, meaning that, in particular, no two guards are positioned within the same corridor.
\end{proof}

Note that this result cannot be directly derived from~\cite{dispersiveAGP-CCCG24}, as it also holds for office-like polygons with holes, while the known result only holds for simple polygons without holes.

\section{Computational complexity} \label{complexity}

In this section, we investigate the computational complexity of the respective decision problem. 
Our main result is the following:

\begin{theorem}
	\label{thm:hardness-integer}
	Deciding whether there exists a guard set with a dispersion distance of~\num{4} in independent office-like polygons with holes and vertices on integer coordinates is \NP-complete.
\end{theorem}

We highlight that the constructed polygon in our reduction is essentially a polyomino, thus, this result also answers an open question from Rieck and Scheffer~\cite{rieck-scheffer-dispersiveAGP}.
Unfortunately, the gadgets are not thin, i.e., those that do not contain a $2\times 2$ square. 
Nevertheless, we can adjust all involved gadgets, showing that this result also holds true for thin polyominoes.

\begin{corollary}
	\label{cor:hardness-polyominoes}
	Deciding whether there exists a guard set with a dispersion distance of~\num{4} in (thin) polyominoes with holes is \NP-complete.
\end{corollary}

By relaxing the constraint on the vertex positions to arbitrary (rational) coordinates, and slight modifications to all involved gadgets, we obtain the following.

\begin{corollary}
	\label{cor:hardness-arbitrary}
	Deciding whether there exists a guard set with a dispersion distance of $2+\varepsilon$ for any rational $\varepsilon > 0$ in independent office-like polygons with holes is \NP-complete.
\end{corollary}

Membership in \NP\ is relatively straightforward: 
First, we can compute the $r$-visibility polygon of a guard in polynomial time, see~\cref{satsolver,app:visibility}. 
Additionally, the union of two polygons can be computed in polynomial time as well~\cite{margalit}.
Thus, we can efficiently verify that a given solution is indeed a guard set.
Moreover, the distance between two guards can be computed in polynomial time~\cite{toth2004handbook}. 
By checking all pairwise distances between guards, we can verify whether the guard set realizes a certain dispersion distance.
Hence, the problem is in \NP.

We~proceed with a high-level idea of the reduction showing \NP-hardness, before moving on to a thorough explanation of the various gadgets that are utilized in the actual construction.

\paragraph*{High-level idea of the reduction}
We reduce from the \NP-complete variant of \textsc{Planar 3Sat} where every variable occurs in exactly three clauses, denoted as \textsc{Planar 3Sat with Exactly Three Occurrences per Variable} (\psat)~\cite{berman,cerioli,garey}.
Without loss of generality, we can assume that each variable appears negated in two clauses and unnegated in one clause.

For any given Boolean formula $\varphi$ representing an instance of \psat, we first compute a rectilinear embedding of the corresponding clause-variable incidence graph. 
Then, we replace each clause vertex with a clause gadget and each variable vertex with a variable gadget, connecting them accordingly. 
(In some variants of our problem, additional auxiliary gadgets are required to construct a valid office-like polygon.)
We can then demonstrate that $\varphi$ is satisfiable if and only if the respectively constructed polygon has a guard set of a certain dispersion distance.

\paragraph*{Gadgets for office-like polygons with vertices at integer coordinates}
\label{par:hardness-scot-integer}
A first observation is that the polygon, that is depicted in~\cref{fig:dist4_banning}, has the property that, if we want to realize a dispersion distance of $4$, either the two blue or the two red guards need to be selected. 
More importantly, this implies that no other guard can be placed within distance~$3$ of the vertices $v$ and $v'$.

\begin{figure}[htb]
    \centering
    \includegraphics[page=5]{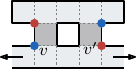}
    \caption{Banning-type subpolygon.}
    \label{fig:dist4_banning}
\end{figure}

\subparagraph{Variable gadget.} 
The variable gadget is depicted in~\cref{fig:4_var_gadget}.
Banning-type gadgets on both sides prevent placing guards at the vertices $f$ and $f'$. 
As a result, no guard set consisting solely of vertices from the variable gadget can see into two consecutive corridors. 
In other words, there exist feasible guard sets that either guard both corridors labeled \texttt{false}, or the corridor labeled \texttt{true}, but never one corridor of each type.
Therefore, this gadget can distinguish between true and false.

\begin{figure}[htb]
	\centering
	\includegraphics[page=1]{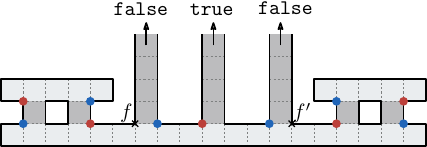}
	\caption{Variable gadget for office-like polygons with vertices at integer coordinates.}
	\label{fig:4_var_gadget}
\end{figure}

\subparagraph{Clause gadget.} Because the problem from which we construct our reduction is solvable in polynomial time in the special case where all clauses contain three literals, we introduce clause gadgets that include both two- and three-literal clauses, as visualized in~\cref{fig:clause-gadgets-integer}.

\begin{figure}[htb]
    \centering
     \begin{subfigure}[b]{0.45\textwidth}
         \centering
         \includegraphics[page=2]{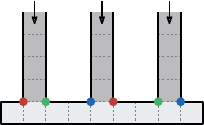}
         \subcaption{$3$-literal clause gadget.}
         \label{fig:4_3clause_gadget}
     \end{subfigure}
     \hfill
     \begin{subfigure}[b]{0.45\textwidth}
         \centering
         \includegraphics[page = 3]{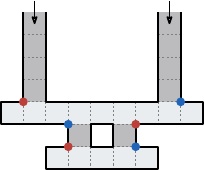}
         \subcaption{$2$-literal clause gadget.}
         \label{fig:4_2clause_gadget}
     \end{subfigure}
    \caption{Clause gadgets for office-like polygons with vertices at integer coordinates.}
    \label{fig:clause-gadgets-integer}
\end{figure}

In both cases, there exist guard sets with a dispersion distance of $4$ that guard all incoming corridors except for one. 
As a result, at least one of the corridors must be guarded by a guard placed on a vertex from another gadget, which implies that the clause is satisfied by the corresponding truth assignment.

To construct an office-like polygon from a given embedding of the clause-variable-incidence graph, it is essential to be able to make $\ang{90}$ turns. 
As these polygons only permit straight corridors, turns can only be realized by auxiliary rooms. 
Note that the banning-type subpolygon in the two-literal clause gadget is necessary only to accommodate the bending of corridors; without it, the corridors would be too close together to provide enough room for the bending gadgets.
For that reason, we now introduce the bending gadget.

\subparagraph{Bending gadget.} The bending gadget is depicted in~\cref{fig:bend-gadget-integer}. 
The banning-type subpolygon prevents placing a guard on the vertex labeled $v$ of the outgoing corridor.
Additionally, the remaining vertex of the outgoing corridor is at a distance of $3$ from the two vertices incident to the incoming corridor. 
Therefore, we can only guard the outgoing corridor from vertices of the bending gadget if the incoming corridor does not need to be guarded, and vice versa, ensuring a feasible propagation of the truth assignment.

\begin{figure}[htb]
	\centering
	\includegraphics[page=4]{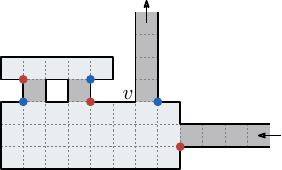}
	\caption{The bending gadget enables \ang{90} turns within office-like polygons whose vertices lie on integer coordinates.}
	\label{fig:bend-gadget-integer}
\end{figure}

\paragraph*{Gadgets for polyominoes}
\label{par:polyomino-gadgets}
As previously noted, this reduction also applies to polyominoes. 
However, we can use even simpler gadgets to demonstrate that determining whether a guard set with a dispersion distance of 4 exists is \NP-hard. 
Moreover, as these gadgets are thin, this result extends to thin polyominoes as visualized in~\cref{fig:hardness-polyominoes}.

\begin{figure}[htb]
    \begin{subfigure}[b]{0.25\textwidth}
         \includegraphics[page = 8]{figures/hardness_polyominoes.pdf}%
         \subcaption{Variable gadget.}
     \end{subfigure}\hfill%
     \begin{subfigure}[b]{0.25\textwidth}
         \includegraphics[page=1]{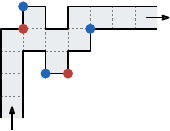}%
         \subcaption{Bending gadget.}
     \end{subfigure}\hfill%
     \begin{subfigure}[b]{0.25\textwidth}
         \includegraphics[page = 4]{figures/hardness_polyominoes.pdf}%
         \subcaption{2-literal clause.}
     \end{subfigure}\hfill%
          \begin{subfigure}[b]{0.25\textwidth}
     	\includegraphics[page = 6]{figures/hardness_polyominoes.pdf}%
     	\subcaption{3-literal clause.}
     \end{subfigure}\hfill%
    \caption
    {Simpler gadgets for polyominoes.
    Note that the 3-clause gadget is shown solely for completeness, as it is equivalent to that for office-like polygons with vertices at integer coordinates.}
    \label{fig:hardness-polyominoes}
\end{figure}

\paragraph*{Gadgets for office-like polygons with vertices at arbitrary rational coordinates}
\label{par:hardness-scot-arbitrary}
Removing the requirement that vertices lie on integer coordinates provides more freedom, enabling a more intricate construction to show hardness distance $2+\varepsilon$ for any $\varepsilon > 0$.
In~particular, corridors can be placed in such a way that they remain independent, yet their vertical or horizontal distance can be arbitrarily small.
The underlying structure of the gadgets for this variant is essentially the same as those in the integer variant, now combined with a polygon similar to the one shown in~\cref{fig:2_sufficient_main}. 
However, as the gadgets grow larger due to the minimum number of nested corridors required for the construction to work properly, an additional gadget is necessary. 
This gadget is used to stretch or widen the respective polygon, creating space for other gadgets.

While the proofs of~\cref{thm:hardness-integer,cor:hardness-polyominoes} are straightforward, the proof of~\cref{cor:hardness-arbitrary} is more involved; all technical details are provided in~\cref{app:complexity-details}.
\section{Dynamic programming for hole-free independent office-like polygons} \label{dp_algo}
Computing optimal solutions for the \dagp turns out to be challenging. 
At~present, the only non-trivial class of polygons known to admit a polynomial-time algorithm comprises polyominoes whose dual graph is a tree~\cite{rieck-scheffer-dispersiveAGP}. 
Additionally, the \NP-hardness result implies that the existence of a polynomial-time algorithm is unlikely, even for office-like polygons that only contain independent corridors.

We now present a polynomial-time dynamic programming algorithm that computes the optimal solution for independent office-like polygons without holes.
The main component of this approach is to solve geometric independent set problems.

\begin{restatable}[$\star$]{theorem}{dpalgo}
    \label{thm:dpalgo}
    For independent office-like polygons without holes, there exists a polynomial-time algorithm that computes guard sets with maximum dispersion.
\end{restatable}

The high-level idea of our approach is as follows: 
Consider a hole-free independent office-like polygon~$\polygon$ with~$n$ vertices. 
We solve the decision problem whether a guard set with a dispersion distance of at least $\ell$ exists. 
Afterward, we do a binary search over the~$\mathcal{O}(n^2)$ many possible dispersion distances (implied by the pairwise vertex distances in $\polygon$) to solve the corresponding maximization problem, i.e., finding a guard set realizing a maximum dispersion distance by using our algorithm for the decision problem as a subroutine.

Let $G(\polygon) = (V,E)$ represent~$\polygon$ as follows: 
There is a node $v \in V$ for each room $R_v$ of~$\polygon$, and an edge $e=\{v,w\} \in E$ if two rooms $R_v$ and~$R_w$ are connected by a corridor $C_e$.
We~now fix an arbitrary in-arborescence $G'$ of $G(\polygon)$, i.e., a directed tree with edges pointing towards a root, see~\cref{fig:overview_dp_outline}. 
For each node $v$ of $G'$, we define a subproblem and solve it only when all predecessor nodes are marked as processed.

We proceed with a high-level overview of the key terminology, followed by precise definitions of the associated subproblems and a detailed exposition of the core components of our algorithm.
Subsequently, we provide a brief analysis of the algorithm's runtime.

\begin{figure}[htb]
    \centering
    \includegraphics[scale = .95]{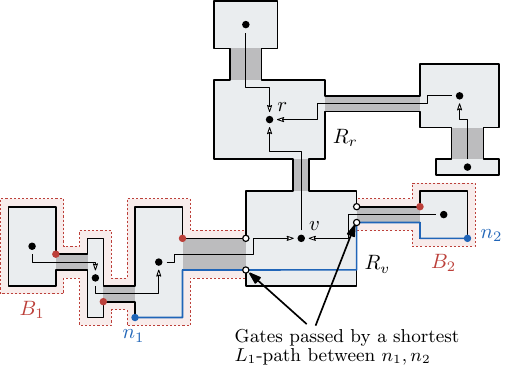}
    \caption{A hole-free independent office-like polygon $\polygon$ with its corresponding in-arborescence for $G(\polygon)$ with root $r$.}
    \label{fig:overview_dp_outline}
\end{figure}

\paragraph*{Preliminaries and subproblem definition}
Consider a directed edge $e=(u,v)$ in $G'$ and the maximal directed subtree~$T$ of~$G'$ rooted at~$u$. 
The~\emph{branch}~$B$ for~$e$ consists of the corridor $C_e$ and all rooms and corridors corresponding to nodes and edges of~$T$. 
We refer to the two vertices shared by~$C_e$ and~$R_v$ as the \emph{gate vertices} (or gate) of~$B$. 
For convenience, we define the branches rooted at~$v$ as those for edges of the form~$(\cdot,v)$. 
In~\cref{fig:overview_dp_outline}, the rooms and corridors outlined in red correspond to the branches $B_1$ and $B_2$ that are rooted at~$v$, with corresponding gate vertices marked in white. 

A guard set for a branch $B$ consists of guards exclusively placed at vertices within~$B$, ensuring coverage of every room and corridor in~$B$. 
However, some guard sets for~$B$ also cover~$R_v$ (by placing a guard at a feasible gate vertex), while others do not.
A set of \emph{configurations}~$\mathcal{C}$ for~$B$ is a collection of guard sets for~$B$, each with a dispersion distance of at least~$\ell$, satisfying the following property:
Consider an arbitrary guard set~$\guardset$ for~$\polygon$ with a dispersion distance of at least~$\ell$. 
Let $\guardset_B \subseteq \guardset$ denote the guards that are placed at vertices within~$B$. 
Then, there always exists a configuration $c \in \mathcal{C}$ such that $(\guardset \setminus \guardset_B) \cup c$ forms a guard set for~$\polygon$ with a dispersion distance of at least~$\ell$.

These definitions allow us to define a subproblem for each node in the in-arborescence~$G'$:
When processing the root~$r$ of $G'$, we solve the entire decision problem. 
Every other node~$v$ in~$G'$, is handled as follows: 
Let~$w$ denote the unique node such that there exists an edge from~$v$ to~$w$ in~$G'$. 
The goal is to find a set of configurations for the branch of $(v,w)$. 
Note that, by the time we process~$v$, all predecessor nodes have already been processed, and we therefore know a set of configurations for each branch rooted at~$v$. 
If, at some point, a computed configuration set is empty, we conclude that no guard set with a dispersion distance of at least~$\ell$ exists for~$\polygon$, not even within this branch.

\paragraph*{Computing configuration sets}
We now explain how we construct the configuration set for a branch $B_e$ belonging to an edge $e = (v,w)$.
Consider all feasible guard placements on the vertices of a single corridor $C_e$ and its incident room $R_v$; as an example, see the red guards in~\cref{configurations_main}.
Note that the number of distinct feasible vertex guard sets in this combined region is clearly bounded by a constant.
\begin{figure}[htb]
	\centering
	\begin{subfigure}[b]{0.18\textwidth}
		\centering
		\includegraphics[page = 1]{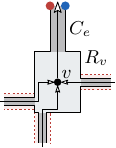}
		\label{conf1}
	\end{subfigure}
	\hfill
	\begin{subfigure}[b]{0.18\textwidth}
		\centering
		\includegraphics[page=2]{figures/configurations.pdf}
	\end{subfigure}
	\hfill
	\begin{subfigure}[b]{0.18\textwidth}
		\centering
		\includegraphics[page = 3]{figures/configurations.pdf}
	\end{subfigure} 
	\hfill
	\begin{subfigure}[b]{0.18\textwidth}
		\centering
		\includegraphics[page = 4]{figures/configurations.pdf}
		\label{corner_place}
	\end{subfigure}
	\hfill
	\begin{subfigure}[b]{0.18\textwidth}
		\centering
		\includegraphics[page = 5]{figures/configurations.pdf}
		\label{no_blue}
	\end{subfigure}
	\caption{Schematic illustration of sensible guard placements at corners of $R_v$ and in $C_e$ (depicted red).
		When fixing one of the options, it suffices to find a guard set for $e$'s branch that has a dispersion distance of at least~$\ell$ and, among those, maximizes the distance to the respective blue vertex.}
	\label{configurations_main}
\end{figure}
Furthermore, some of these guard sets are not sensible, as, e.g., placing two guards at corners of a room $R_v$ is unnecessary, as both would have identical visibility polygons.
Now, for every such placement~$X$, we aim to derive a configuration for~$B_e$ by selecting the guards in~$X$ along with one precomputed configuration for each branch rooted at~$v$.
However, choosing such a combination requires careful consideration of the following aspects:

\descriptionlabel{(1)} 
As mentioned above, some configurations for branches rooted at~$v$ do cover $R_v$, while others do not.
Therefore, if~$X$ does not contain a guard in $R_v$, an arbitrary choice of configurations may result in $R_v$ remaining uncovered.
To address this issue, we fix a configuration (for some branch rooted at~$v$) that places a guard at one of its gate vertices, and hence in $R_v$.
To ensure correctness, we probe every possible option.

\medskip
\descriptionlabel{(2)} 
By definition, all pairs of guards within a single configuration~$c$ maintain a minimum dispersion distance of~$\ell$; however, this condition may be violated when considering interactions between guards in~$X$ and those in~$c$.
In such instances, the configuration~$c$ is deemed invalid and consequently excluded from further consideration.

\medskip
\descriptionlabel{(3)}
When computing the configuration~$c$ that contains the guard set $X$, our goal is to ensure that whenever there exists a guard set for~$\polygon$ with a dispersion distance of at least~$\ell$ that includes~$X$, then there also exists such a guard set that includes~$c$. 
Leveraging on the fact that $\polygon$ is independent, we arrive at the following crucial observation.

\begin{restatable}{observation}{gateVerticesShortestPath}\label{gate_dist_main}
    Consider two vertices~$n_1, n_2$ in~$\polygon$ such that~$n_1$ is inside, and $n_2$ is outside of a branch~$B$.
    Then, there always exists a shortest geodesic $L_1$-path between~$n_1$ and~$n_2$ that passes through one of the gate vertices of~$B$.
\end{restatable}

An example for~\cref{gate_dist_main} is highlighted by the blue paths in~\cref{fig:overview_dp_outline}. 
Moreover, again due to the assumption that $\polygon$ is independent, at least one guard needs to be placed in every corridor, and hence in $C_e$.
Therefore, when considering a fixed placement~$X$, the following holds for one gate vertex~$g$ of~$B_e$:  
when all guards in the chosen configuration maximize their minimum distance to~$g$, the distance to the other gate vertex of~$B_e$ is maximized as well.
For an exemplary illustration, consider the blue vertices in \cref{configurations_main}.
With \cref{gate_dist_main}, this is best possible when using~$X$.

We now want to choose a configuration with these properties. 
To achieve this, we perform a binary search over the $\mathcal{O}(n^2)$ pairwise vertex distances in~$\polygon$, aiming to find the choice of configurations for branches rooted at~$v$ that maximizes the distance to~$g$. 
For each distance~$\ell'$ considered during the current step of the binary search, we discard all configurations of branches rooted at~$v$ where a guard has a distance less than~$\ell'$ to $g$.

\medskip
\descriptionlabel{(4)} 
After addressing the aforementioned aspects, the crucial task is to select one configuration for each branch rooted at~$v$ so that no two guards from different configurations are closer than~$\ell$, or to determine that no such selection exists.

This problem corresponds to finding an independent set of size~$k$ in the following auxiliary graph, where~$k$ denotes the number of branches rooted at~$v$:
Each configuration of a branch rooted at~$v$ is represented by a node. 
Two nodes are connected by an edge if both corresponding configurations $c_1, c_2$ belong to the same branch, or if there exists a guard $g_1 \in c_1$ and a guard $g_2 \in c_2$ such that the distance between them is strictly smaller than $\ell$.

To this end, we propose a straightforward algorithm that iteratively prunes infeasible configurations--those that can never contribute to a valid solution--proceeding until a solution is identified or the algorithm certifies that no solution exists.
The algorithm essentially consists of two phases.
In the first phase, the algorithm selects configurations (with sufficiently large dispersion distance) in two steps for all branches geometrically extending (1)~vertically and (2)~horizontally from the room $R_v$, respectively.
This can be done using a simple greedy approach.
In the second phase, the algorithm combines the solutions for both subproblems, either finding a valid solution, or identifying a configuration that can provably be discarded.
In the latter case, the algorithm is restarted with the updated set of available configurations.

\paragraph*{Runtime of the dynamic programming approach}
To solve the decision problem, we need to solve one subproblem for each node of $G'$. 
The~binary search on the possible dispersion distances to solve the maximization problem takes an additional logarithmic factor.
This yields a total runtime of $\mathcal{O}(t_n \cdot n \log n)$, where~$t_n$ denotes the time needed to solve a subproblem. 
We provide a detailed description of the algorithm for solving a subproblem and prove that its amortized runtime across all subproblem calls is $\mathcal{O}(n^4 \log n)$ in~\cref{alg_subproblems}.
With this, the total runtime of the dynamic programming approach is $\mathcal{O}(n^5 \log^2 n)$. 

Full details of all steps of our approach are given in~\cref{dp_algo_x}.
\section{Computing optimal solutions in practice} \label{satsolver}

Mixed Integer Programming is a popular approach for solving \NP-hard optimization problems.
This problem allows for a simple formulation.
We define a binary variable $x_g\in \mathbb{B}$ for each vertex $g\in V(\mathcal{P})$ in polygon $\mathcal{P}$ to indicate guard placement and a continuous variable $\ell \in \mathbb{R}^+_0$ for the dispersion distance.
The objective maximizes $\ell$, subject to (i) $\forall w\in \mathcal{W}: \sum_{g \in V(\mathcal{P}), w\in Vis(g)} x_g \geq 1$, ensuring every witness in $\mathcal{W}$ is covered, and (ii) $\forall g, g' \in V(\mathcal{P}): x_{g} \wedge x_{g'} \rightarrow \ell \leq \delta(g, g')$, enforcing that no two guards $g,g'\in V(\mathcal{P})$ with distance $\delta(g, g')$ smaller than $\ell$ can be selected.
The second constraint is non-linear but can be reformulated into a linear constraint using the big-M method.
However, this reformulation can be problematic for traditional MIP solvers, as it weakens the linear relaxation and may lead to numerical instabilities.
Alternatively, constraint programming solvers like CP-SAT can handle these constraints more effectively, as they rely less on linear relaxation.

Less obvious is that a SAT solver can also be used to solve this problem, leveraging the fact that there are only $\mathcal{O}(n^2)$ possible objective values to check for feasibility.
For a fixed $\ell$, feasibility is determined by the constraints $\bigwedge_{w \in \mathcal{W}} \left(\bigvee_{g \in V(\mathcal{P}), w \in \text{Vis}(g)} x_g\right)$, ensuring coverage, and $\bigwedge_{g, g' \in V(\mathcal{P}), \delta(g, g') < \ell} \left(\overline{x_{g}} \vee \overline{x_{g'}}\right)$, enforcing the minimum guard distance.
A binary search over $\ell$ then yields the optimal solution.
For efficiency, $\ell$ should be updated based on the actual solution returned by the SAT solver rather than just the probed values.

To obtain a sufficient witness set $\mathcal{W}$, we use atomic visibility polygons (AVPs) as introduced by Couto et al.~\cite{couto}, which yield a shadow witness set.
The key idea is to construct the arrangement of visibility polygons defined by the polygon's vertices, where each face (referred to as AVP) is covered by the same guard set.
From this arrangement, shadow AVPs are identified as the local minima in the partial order of AVPs based on their covering sets.
Selecting a single witness from each shadow AVP ensures full coverage of the polygon, even under $r$-visibility constraints.
For computing the $r$-visibility polygons, we implemented a simple $\mathcal{O}(n \log n)$ algorithm; see~\cref{app:visibility}. 
The basic idea is to partition the visibility region into four quadrants. 
Within each quadrant, we compute a Pareto front over the vertices, which primarily requires sorting followed by a sequential scan.

\subparagraph{Empirical evaluation.}\label{experiments}
We analyze the practical complexity\footnote{Code and data can be found at \url{https://github.com/KaiKobbe/dispersive_agp_solver}.} by addressing the following research questions:

\begin{description}
    \item[RQ1] Among the available SAT solvers, which one performs best for our problem?
    \item[RQ2] Which approach -- MIP, CP-SAT, or SAT -- yields the best performance?
    \item[RQ3] How does the runtime vary across different types of orthogonal polygons?
    \item[RQ4] How is the total runtime distributed among the individual computational steps?
\end{description}

To answer these questions, we use a benchmark of \num{2344} instances, consisting of \num{1599} randomly generated office-like polygons and \num{745} orthogonal polygons from the SBGDB~\cite{Ede&20d}, with up to \num{1600} vertices.
The office-like polygons, generated with and without holes, are created by iteratively placing randomly sized rooms in the plane, ensuring connectivity through corridors.
Additional corridors are added to introduce holes where applicable.
We generated \num{20} instances for each multiple of \num{40} vertices, ranging from \num{40} to \num{1600}.

All instances were executed once for each solver on an Ubuntu Linux workstation equipped with an AMD Ryzen 9 7900 processor and \SI{96}{\giga\byte} of RAM.
The implementation was written in Python 3.12.8 using PySAT 1.8dev14, OR-Tools 9.11.4210, and Gurobi 12.0.1.
Geometric computations were performed in C++ using CGAL 5.6.1 and integrated via PyBind11.

\begin{description}
\item[RQ1] \Cref{fig:satplt}~(top) indicates that all SAT solvers perform similarly, with Glucose4 demonstrating a slight performance advantage. 
Consequently, we select Glucose4 as the default SAT solver for the subsequent experiments.
\item[RQ2]
\Cref{fig:satplt}~(bottom) demonstrates that the SAT-based approach significantly outperforms the MIP and CP-SAT approaches, requiring only a fraction of the time, particularly for orthogonal polygons.
While CP-SAT was still able to solve all instances within \SI{300}{\second}, Gurobi failed to solve \num{30} instances, despite being faster for smaller instances.
\item[RQ3] While MIP and CP-SAT can handle office-like polygons significantly faster than orthogonal polygons, the SAT-based approach shows no significant runtime difference between polygon types.
Overall, neither holes nor the type of orthogonal polygon substantially impact the runtime of the SAT-based approach.
However, these factors do influence the MIP and CP-SAT approaches visibly.
\item[RQ4] For instances with at least \num{1500} vertices, the SAT-based approach spends on average \SI{4.397}{\second} of the overall runtime of \SI{10.475}{\second} on building the witness sets.
On average, \num{14.69} probes of the objective value are required, of which \num{11.56} result in infeasibility.
The~time spent in the SAT solver is only \SI{0.011}{\second} on average, with the remaining time spent on building the models.
\end{description}

\begin{figure}[htb]
    \centering
    \begin{subfigure}[b]{0.9\textwidth}
        \centering
        \includegraphics[width=0.675\columnwidth]{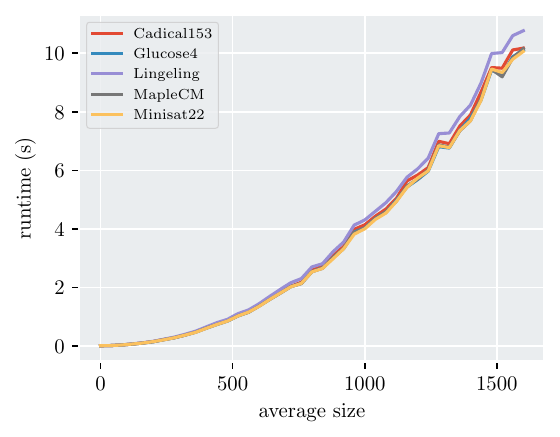}
    \end{subfigure}\\
    \vspace*{1em}
    \begin{subfigure}[b]{0.9\textwidth}
         \centering
         \includegraphics[width=0.675\columnwidth]{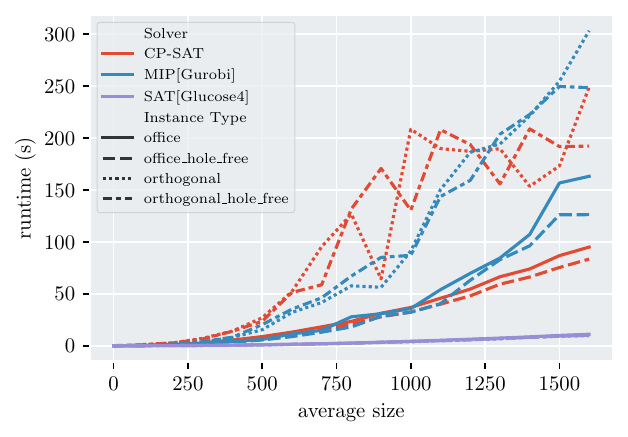}
     \end{subfigure}
    \caption{(top) Comparison of the different SAT solvers. (bottom) Comparison of the runtime (limited to \SI{300}{\second}) of the different approaches on different instance types.}
    \label{fig:satplt}
\end{figure}	
\section{Conclusions}\label{sec:conclusions}
We conducted an extensive investigation of the \dagp within orthogonal polygons, focusing specifically on office-like polygons, and established a range of results. 
To this end, we distinguished between cases where polygon vertices are restricted to integer coordinates and those where they are not.

On the theoretical side, we proved that the problem is \NP-complete, even for independent office-like polygons. 
In particular, we showed hardness of deciding whether there exists a guard set that realizes a dispersion distance of $4$ in office-like polygons with vertices at integer coordinates.
The construction also shows that the same holds true for polyominoes.
Complementarily, we presented polynomial-time algorithms for computing worst-case optimal solutions of dispersion distance $3$.
For independent office-like polygons without holes, we further introduced a dynamic programming approach that efficiently computes the optimal solution in polynomial time.
It~remains an open question whether the dynamic programming approach can be extended to accommodate the case where corridors are not independent.
Moreover, the question of whether guard sets with a dispersion distance of~$3$ can be computed in polynomial time for polyominoes containing holes remains unresolved.

On the practical side, despite the problem's theoretical complexity, we demonstrated its practical tractability for arbitrary orthogonal polygons under the constraint of $r$-visibility, at least for random instances up to a size of \num{1600} vertices.

While research has focused on the problem with vertex guards, the variant involving point guards remains an open question.
   	\bibliography{bibliography}

	\appendix
	\section{Omitted details of~\cref{chap:wc_opt}: Worst-case optimality}
We provide omitted details for \cref{chap:wc_opt}.
We start by providing a thorough analysis of a worst-case optimal algorithm that guarantees a dispersion distance of at least~$3$ for office-like polygons with vertices at integer coordinates.
Afterward, for the unrestricted case, we construct office-like polygons in which every guard set has a dispersion distance smaller than~$2+\varepsilon$ for any~$\varepsilon > 0$. 

\subsection{Proof of \cref{thm:worst-case-integer-scots}}
\label{app:worst-case-integer}
\newcommand{\vbl}[0]{v_i^{bl}}
\newcommand{\vtl}[0]{v_i^{tl}}
In \cref{vpic_main}, we already discussed examples of office-like polygons with vertices at integer coordinates for which a dispersion distance of~$3$ is optimal.
We now provide a detailed description and analysis of an algorithm that computes worst-case optimal guard sets.

\worstCaseIntegerScots*

\begin{proof}
    Consider an office-like polygon~$\polygon = (\mathcal R, \mathcal C)$ with vertices at integer coordinates. 
    We construct a guard set for~$\polygon$ with a dispersion distance of at least~$3$ as follows. 
    First, we place guards to cover all vertical corridors. 
    Second, we basically repeat this procedure to cover all horizontal corridors. 
    Third, we place guards at corners of rooms that are not visible to any previously placed guard. 
    
    The union of these guards forms a guard set for~$\polygon$. 
    Moreover, the algorithm never places a guard with a distance less than~$3$ to a previously placed guard.
    As an example, \cref{fig:example_wcopt_algo} illustrates the guard set obtained by the algorithm.

    \begin{figure}[htb]
    \centering
    \includegraphics[page = 2]{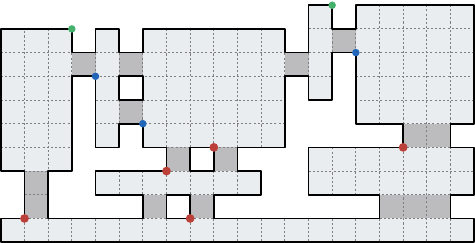}
    \caption{Exemplary office-like polygon with guard set computed by our worst-case optimal algorithm. Red guards are added in the first phase, blue guards in the second phase, and green guards in the third phase.}
    \label{fig:example_wcopt_algo}
    \end{figure}

    \subparagraph{Phase (1): Vertical corridors.} 
    We define a partial order on~$\mathcal R$ such that, for every pair of rooms $R_1, R_2 \in \mathcal R$, we say~$R_1 \prec R_2$ if there is a vertical corridor incident to both, and~$R_1$ is below~$R_2$. 
    The general approach is to select a room~$R$ that is a local maximum among the unprocessed rooms, process it, and then continue with a next local maximum. 
    More precisely, when processing~$R$, we place guards which ensure that all corridors extending to the top from~$R$ are seen. 
    We do so by exclusively placing guards on left walls of these corridors.

    Following, we explain how to place guards when processing room~$R$. 
    Let~$C_1, \dots, C_j$ denote the corridors extending to the top from~$R$, ordered from left to right. 
    We consider~$C_i$ for increasing~$i$, and denote by~$\vbl$ (or~$\vtl$) the bottom-left (or top-left) vertex of~$C_i$. 
    The algorithm works as follows: 

    \begin{enumerate}
        \item If~$C_i$ is seen by a previously placed guard, then do not place a guard in~$C_i$.
        \item Else, if no guard placed so far has distance less than~$3$ to~$\vbl$, then place a guard at~$\vbl$.
        \item Else, place a guard at~$\vtl$.
    \end{enumerate}
        
    Clearly, after the~$j$-th step, the corridors~$C_1, \dots, C_j$ are seen. 
    We analyze the strategy more deeply in order to show that each time the algorithm places a guard~$g$ in corridor~$C_i$, no guard placed so far has a distance of less than~$3$ to~$g$. 
    This is true by construction if a guard is placed at~$\vbl$. We show that the same holds if a guard is placed at~$\vtl$. 
    First, we establish some useful invariants.

    \begin{claim}\label{v1}
        At the moment where the algorithm places a guard at~$\vtl$, it holds that
        \begin{itemize}
            \item $i \geq 2$, i.e., there is another corridor to the left of~$C_i$ that is extending to the top from~$R$,
            \item $\delta(v^{bl}_{i-1}, \vbl) = 2$, and
            \item a guard is placed at~$v^{bl}_{i-1}$.
        \end{itemize}
    \end{claim}

    \begin{claimproof}
        Recall that the algorithm places a guard at~$\vtl$ only if some already placed guard has a distance of at most~$2$ to~$\vbl$. 
        To this end, we identify possible guard positions in distance~$2$ from~$\vbl$. 
        For a schematic illustration, consider~\cref{fig:disp3v1}, where the dark-red point corresponds to~$\vbl$, and the other colored points indicate all positions (at integer coordinates) where a guard would be in distance~$2$ from~$\vbl$. 
        The rooms and corridors may appear slightly different than depicted, but the illustration is intended to provide orientation for the positions. 
        Moreover, note that not all depicted positions must correspond to actual vertext positions. 
        We observe the following:

        First, no guard is placed at a green position because, so-far, no guard has been placed at another vertex in~$C_i$, in a corridor~$C_k$ for~$k > i$, at a corner of a room, or in any corridor extending to the bottom from~$R$. 
        Second, no guard is placed at a turquoise position because such a guard would see~$C_i$, and hence the algorithm would not place a guard at~$\vtl$, a contradiction to the assumption. 
        Third, no guard is placed at any of the light-red positions, as these (assuming they are vertex positions) must be at a corner of a room or on the right wall of~$C_{i-1}$. 
        Recall that we only place guards on left walls of vertical corridors. 
        Therefore, the only possible guard position remaining is indicated in blue. 
        Because there must be a guard in distance~$2$ from~$\vbl$, it follows that a guard is placed exactly at the blue position, which must therefore be a vertex. 
        It is easy to verify that the blue position corresponds to~$v_{i-1}^{bl}$. 
        This also certifies that~$i \geq 2$.
    \end{claimproof}
        
    Building on~\cref{v1}, we show that if the algorithm places a guard at~$\vtl$, then this is a reasonable decision.
        
    \begin{claim}\label{v2}
        If the algorithm places a guard at~$\vtl$, then no previously placed guard has a distance of less than~$3$ to~$\vtl$.
    \end{claim}

    \begin{claimproof}
        For a schematic illustration, consider~\cref{fig:disp3v2}, where the dark-red point corresponds to~$\vtl$. 
        The other colored points indicate all positions (at integer coordinates) where a guard would be in distance~$2$ from~$\vtl$. 
        
        Similar to the above proof, no guard is currently placed at a green position. 
        Moreover, no guard is placed at a turquoise position because such a guard would see~$C_i$, and hence the algorithm would not place a guard at~$\vtl$, a contradiction. 
        
        From \cref{v1}, it follows that a guard must be placed at~$v_{i-1}^{bl}$, and hence no guard is placed at one of the light-red positions; note that we place at most one guard in each corridor. 
        
        Additionally, a guard at the magenta position would see~$C_{i-1}$ and hence no guard would have been placed at~$v_{i-1}^{bl}$, a contradiction to \cref{v1}. 
        Here we use the fact that~$\delta(\vbl, v^{bl}_{i-1}) = 2$, and that every corridor has a width of at least~$1$. 
        
        It remains to be shown that no guard is placed at the blue position. 
        If the blue position is in the same room~$R'$ as~$\vtl$ (i.e., $R'$ has height~$2$), then a guard at the blue position would also see~$C_i$, and hence the algorithm would not place a guard at~$\vtl$, a contradiction. 
        Following, we assume that~$R'$ has height~$1$, because otherwise the blue position can never correspond to a vertex.
        To this end, consider \cref{fig:disp3v3}, where~$w_k^{tl}$ corresponds to the blue position in \cref{fig:disp3v2}.
        Following \cref{v1}, if a guard is placed at~$\vtl$, then there is also a guard at~$v^{bl}_{i-1}$. 
        An analogous statement can be found for a guard at~$w^{tl}_k$, whose placement implies that another guard must be at~$w^{bl}_{k-1}$. 
        These two guards have been placed in a previous iteration, when the room~$R'$ was a local maximum. 
        Since~$2 = \delta(\vbl, v^{bl}_{i-1}) = \delta(w^{bl}_k, w^{bl}_{k-1})$, and because~$\vtl$ and~$w^{tl}_k$ lie on the same vertical line, it follows that the guard at~$w^{bl}_{k-1}$ sees~$C_{i-1}$. 
        As a consequence, no guard will be placed in~$C_{i-1}$, and hence at~$v^{bl}_{i-1}$, a contradiction.

        We have shown that no previously placed guard is in distance~$2$ from~$\vtl$, which proves the statement.
    \end{claimproof}

    \begin{figure}[htb]
        \centering
        \begin{subfigure}[b]{0.45\textwidth}
            \centering
            \includegraphics[page=1]{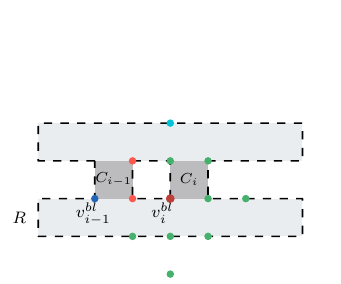}
            \subcaption{}
            \label{fig:disp3v1}
        \end{subfigure}
        \hfill
        \begin{subfigure}[b]{0.45\textwidth}
            \centering
            \includegraphics[page=2]{figures/disp3_distances.pdf}
            \subcaption{}
            \label{fig:disp3v2}
        \end{subfigure}
        \hfill
        \begin{subfigure}[b]{0.45\textwidth}
            \centering
            \includegraphics[page=3]{figures/disp3_distances.pdf}
            \subcaption{}
            \label{fig:disp3v3}
        \end{subfigure}
        \caption{Schematic illustrations for the proof of \cref{v1,v2}.}
        \label{disp3_distances}
    \end{figure}

    This concludes the first phase. 
    Clearly, after processing all rooms, every vertical corridor is seen by a guard, all guards are placed exclusively on left walls of vertical corridors, and no two guards have a distance of less than~$3$. 
    For an exemplary illustration, the red guards in \cref{fig:example_wcopt_algo} correspond to the guards placed up to this point.

    \subparagraph{Phase (2): Horizontal corridors.} 
    Similar to the procedure for vertical corridors, we define a partial order on~$\mathcal R$ such that~$R_1 \prec R_2$ if both rooms are connected by a horizontal corridor, and~$R_1$ is to the right of~$R_2$. 
    We apply a similar approach as before, but this time, we place guards exclusively on lower walls of horizontal corridors.
    One could view the second phase as rotating~$\polygon$ by \ang{90} clockwise and then rerunning the first phase.
    Therefore, we omit further details.
    With similar reasoning as in the first phase, one can show that every pair of guards placed in the second phase has a distance of at least~$3$. 
    It remains to be shown that no guard placed in the second phase has a distance smaller than~$3$ to a guard placed in the first phase.
    We begin by proving an auxiliary claim.

    \begin{claim}\label{aux_step2}
        Let~$C$ denote the leftmost corridor extending to the bottom from any room~$R$. The algorithm will never place a guard at the top-left vertex of~$C$.
    \end{claim}

    \begin{claimproof}
        Towards a contradiction, assume the algorithm places a guard $g_1$ at the top-left vertex of~$C$.  
        Let~$R'$ denote the room connected to~$R$ through~$C$. 
        Note that~$g_1$ was placed during the first phase of the algorithm, specifically in the iteration where~$R'$ was a local maximum. 
        Let~$C_1, \dots, C_j$ denote the corridors extending to the top from~$R'$, from left to right. 
        Moreover, let~$C \equiv C_i$.
        For an exemplary illustration, consider \cref{fig:case3_union}.
        
        Returning to the previous notation from the first phase, the guard~$g_1$ is placed at the vertex~$\vtl$. 
        Observe that either~$i=1$, or else~$\delta(v_{i-1}^{bl}, v_i^{bl}) \geq 3$, because~$C_{i-1}$ leads to a different room; see \cref{fig:case3_union}. 
        For both these cases, it follows from \cref{v1} that the algorithm will not place~$g_1$ at~$\vtl$, i.e., at the top-left vertex of~$C_i$.
    \end{claimproof}

    Equipped with \cref{aux_step2}, we now show that distances between guards are sufficiently large after the second phase.

    \begin{claim}\label{claim:step2_dist}
        No guard placed in the first phase has a distance that is smaller than~$3$ to a guard placed in the second phase.
    \end{claim}
    
    \begin{claimproof}
        Consider two guards $g_1, g_2$ where~$g_1$ is placed during the first phase (i.e., in a vertical corridor), and~$g_2$ is placed during the second phase (i.e., in a horizontal corridor). 
        Without loss of generality, assume that~$g_1,g_2$ are placed in the same room~$R$, as otherwise their distance clearly is as least~$3$.
        We distinguish three cases.
    
        First, assume the corridor in which~$g_2$ is placed extends to the right from~$R$. 
        As shown in \cref{fig:case1_union}, we observe that~$g_2$ has a distance of at least~$3$ to all guards in corridors extending to the top or bottom from~$R$. 
        Recall that guards in vertical corridors are exclusively placed on left walls. 
        Hence, for the remainder, assume the corridor in which~$g_2$ is placed extends to the left from~$R$.

        Second, assume~$g_1$ is placed in a corridor extending to the top from~$R$. 
        Because~$g_2$ is placed on the lower wall of its respective corridor, it is again easy to see that the distance is at least~$3$; see~\cref{fig:case2_union}.
    
        Third, assume~$g_1$ is placed in a corridor extending to the bottom from~$R$. 
        Note that the only possible scenario where~$g_1$ and~$g_2$ could be in distance~$2$ is shown in \cref{fig:case3_union}: 
        both guards have a distance of~$1$ to the bottom-left corner of~$R$.
        This implies that~$g_1$ is placed at the top-left vertex of the leftmost corridor that extends to the bottom from~$R$. 
        However, as proved in \cref{aux_step2}, this case will not occur.

        We have covered all possible cases, and hence we get~$\delta(g_1, g_2) \geq 3$.
    \end{claimproof}
    
    In summary, all guards placed so far ensure that every corridor in~$\polygon$ is seen.
    Moreover, no two guards have a distance that is smaller than~$3$.
    The blue guards in \cref{fig:example_wcopt_algo} were added during the second phase.

    \begin{figure}[htb]
    \centering
    \begin{subfigure}[b]{0.3\textwidth}
        \includegraphics[page=3]{figures/union_horizontal_vertical.pdf}
        \subcaption{}
        \label{fig:case1_union}
    \end{subfigure}
    \hfill
    \begin{subfigure}[b]{0.3\textwidth}
        \includegraphics[page=1]{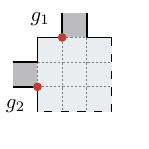}
        \subcaption{}
        \label{fig:case2_union}
    \end{subfigure}
    \hfill
    \begin{subfigure}[b]{0.3\textwidth}
         \includegraphics[page = 2]{figures/union_horizontal_vertical.pdf}
         \subcaption{}
         \label{fig:case3_union}
    \end{subfigure}
    \caption{Exemplary illustrations for the proof of \cref{aux_step2,claim:step2_dist}.}
    \label{union_figures}
\end{figure}

    \subparagraph{Phase (3): Rooms.} 
    In this phase, we place an additional guard in every room that has not been guarded so far.
    Specifically, we place a guard at the respective top-right corner of these rooms. 
    Since guards in vertical (or horizontal) corridors are exclusively placed on left (or bottom) walls, it is easy to see that these newly placed guards always have a distance of at least~$3$ to all others. 
    The green guards in \cref{fig:example_wcopt_algo} are added during this phase.

    This concludes the algorithm and its proof.
\end{proof}

\subsection{Packed corridors and proof of \cref{lem:disp_2eps_example}}
The main result of this subsection is as follows.

\disptwoeps*

To prove \cref{lem:disp_2eps_example}, we introduce a family of (partial) polygons that are office-like and force small dispersion distances through densely packed corridors.

\begin{definition}[$c$-$\varepsilon$-packing] 
    For $c \in \mathbb{N}$ and $\varepsilon \in (0,1)$, a $c$-$\varepsilon$-packing $\polygon'$ is a subpolygon of an office-like polygon~$\polygon$, consisting of three long rooms placed below each other, and~$c$ many unit-square corridors alternatingly placed between the top and bottom room.
    We refer to \cref{fig:wcopt_twoeps} for an illustration.
    More precisely, $c$-$\varepsilon$-packings are constructed as follows. 
    \begin{itemize}
        \item For every $0 \leq i \leq c-1$, place a unit-square corridor with center $(i + i \cdot \tau, 2)$ if $i$ is even, and $(i + i \cdot \tau , 4)$ if~$i$ is odd.
            In this regard, $\tau$ can be chosen arbitrarily in range $(0, \varepsilon /2)$.
        \item Place three rooms such that the first room is incident to all corridors with center $(\cdot, 2)$, the second room is incident to all corridors, and the third room is incident to all corridors with center $(\cdot, 4)$.
            Note that this implies the second room has a height of~$1$. Aside from these constraints, the rooms may have arbitrary size.
    \end{itemize}
    In \cref{fig:wcopt_twoeps}, the red rectangles indicate the $\tau$-spacing between corridors. 
    We require that every corridor in $\polygon'$ is independent of every other corridor in~$\polygon$.
\end{definition}

Following, we show that no $11$-$\varepsilon$-packing hosts a guard set with a dispersion distance of at least~$2+\varepsilon$. 
It is then easy to construct office-like polygons with arbitrarily many vertices, for which every guard set has a dispersion distance of at most~$2+\varepsilon$.
Even though this is sufficient to prove \cref{lem:disp_2eps_example}, we also consider $9$-$\varepsilon$-packings and $10$-$\varepsilon$-packings, as they will be useful later to prove \NP-hardness.

\begin{lemma}
    For $11$-$\varepsilon$-packings, there is no guard set with a dispersion distance of at least~$2 + \varepsilon$.
\end{lemma}

\begin{proof}
    Consider some fixed~$\varepsilon \in (0,1)$.
    For a schematic illustration, refer to the $11$-$\varepsilon$-packing~$\polygon$ depicted in \cref{fig:wcopt_twoeps}.
    Recall that all corridors are unit-squares, the middle room has height~$1$, and the red rectangles have width~$0 < \tau < \varepsilon / 2$.
    Since all corridors are independent, every guard set for~$\polygon$ must place at least one guard in each corridor.

    Towards a contradiction, assume that there exists a guard set~$\guardset$ for~$\polygon$ with a dispersion distance of at least~$2+\varepsilon$.
    For symmetry reasons, we can assume that~$\guardset$ contains a guard at a vertex on the left wall of corridor~$C$.
    However, if~$\guardset$ contains a guard at the green (or blue) vertex of~$C$, then all guards depicted green (or blue) must be included in~$\guardset$:
    In \cref{fig:wcopt_twoeps}, the incoming arcs of every guard~$g$ indicate a set of predecessor guards that leave only~$g$ as guard position in the corridor where~$g$ is placed, unless implying a dispersion distance of at most~$2+2\tau < 2 + \varepsilon$.
    Note that the graph is cycle-free and that the only guard without predecessor is the one placed in~$C$, respectively.

    To finish the proof, observe that if~$\guardset$ contains all the green (or blue) guards, then each vertex in $C_2$ (or $C_3$) has a distance of at most $2+2\tau < 2 + \varepsilon$ to some guard in~$\guardset$.
    This implies that~$\guardset$ does not contain a guard in every corridor, what yields the contradiction to~$\guardset$ being a guard set.
\end{proof}

We make use of the above argumentation to reason about guard sets in $10$-$\varepsilon$-packings.

\begin{corollary} \label{10-piece}
    For $10$-$\varepsilon$-packings, there exists exactly one guard set with a dispersion distance of at least~$2 + \varepsilon$.
\end{corollary}

\begin{proof}
    Reconsider the $11$-$\varepsilon$-packing depicted in \cref{fig:wcopt_twoeps}.
    We obtain a $10$-$\varepsilon$-packing by deleting $C_3$.
    The green guards witness that every guard set containing the bottom-left vertex (or by symmetry the bottom-right vertex) of $C$ has a dispersion distance that is smaller than~$2+\varepsilon$.
    The blue guards witness the existence of a unique guard set having a dispersion distance of at least~$2+\varepsilon$ among all guard sets that place a guard at the top-left vertex of~$C$.
    Moreover, it is easy to see that every guard set containing the top-right vertex of~$C$ will have a dispersion distance smaller than~$2+\varepsilon$. 
    (One could construct a witnessing family of guards similar to the blue one by starting with a guard at the top-right vertex of~$C$, and then propagating enforced guard positions. In this case, $C_1$ will remain uncovered.)
\end{proof}

Similarly, we investigate $9$-$\varepsilon$-packings.

\begin{corollary} \label{9-piece}
    For $9$-$\varepsilon$-packings, there exist exactly two guard sets with a dispersion distance of at least~$2 + \varepsilon$.
\end{corollary}

\begin{proof}
    Reconsider the $11$-$\varepsilon$-packing depicted in \cref{fig:wcopt_twoeps}.
    We obtain a $9$-$\varepsilon$-packing by deleting $C_1$ and $C_3$.
    The green guards witness that every guard set containing the bottom-left vertex of $C$ has a dispersion distance that is smaller than~$2+\varepsilon$.
    The blue guards (except for the one in~$C_1$) witness the existence of a unique guard set having a dispersion distance of at least~$2+\varepsilon$ among all guard sets placing a guard at the top-left vertex of~$C$.
    For symmetry reasons, there is no guard set containing the bottom-right vertex and exactly one guard set containing the top-right vertex of~$C$, while having a sufficient dispersion distance.
\end{proof}

\begin{figure}[htb]
    \centering
    \includegraphics[page=2]{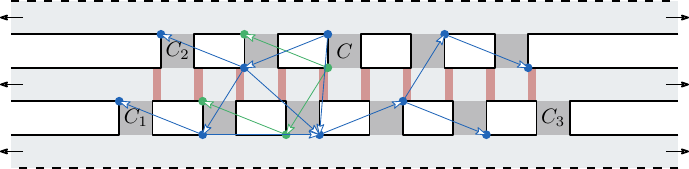}
    \caption{Schematic illustration of an $11$-$\varepsilon$-packing. Every guard set has a dispersion distance that is smaller than~$2+\varepsilon$.}
    \label{fig:wcopt_twoeps}
\end{figure}
	\section{Omitted details of~\cref{complexity}: Computational complexity}
\label{app:complexity-details}

In this section, we provide the complete technical details and any previously omitted content related to the computational complexity of the problem.

\subparagraph{\textsc{Planar 3SAT With Exactly Three Occurrences per Variable} (P3SAT$_{\overline{\underline{3}}}$).} An instance of \textsc{3Sat} consists of a set of boolean variables $\mathcal X$ and a set of clauses $\mathcal C$. For a variable $x \in \mathcal X$, we say $x$ and $\overline{x}$ are literals over $\mathcal X$ such that $x$ is \texttt{true} if and only if $\overline{x}$ is \texttt{false}. 
A clause $C \in \mathcal C$ is a set of either two or three literals over $\mathcal X$. An instance of \textsc{3Sat} is \emph{planar} if the bipartite \emph{clause-variable incidence graph} $G = (\mathcal X \cup \mathcal C, E)$ is planar, where $\{x, C\} \in E$ if and only if $\{x, \overline{x}\} \cap C \neq \emptyset$. 
Moreover, an instance of \psat is an instance of \textsc{Planar 3Sat} where every variable occurs exactly once as positive literal and twice as negative literal in the clauses.

A clause~$C$ is satisfied if, for a given truth assignment to $\mathcal X$, at least one of the literals in~$C$ evaluates to \texttt{true}. 
The problem \psat asks for a truth assignment to $\mathcal X$ such that every clause is satisfied.

We use the following instance of \psat for exemplary illustrations. 
\begin{itemize}
	\item Variables $V = \{x_1, x_2, x_3, x_4\}$
	\item Clauses $\mathcal C = \{C_1, C_2, C_3, C_4\}$, with $C_1 := \{x_1, x_2, \overline{x_3}\}$, $C_2 := \{\overline{x_1}, \overline{x_2}, x_4\}$,\\ ${C_3 := \{\overline{x_1}, x_3, \overline{x_4}\}}$, and $C_4 := \{\overline{x_2}, \overline{x_3}, \overline{x_4}\}$
\end{itemize}
A satisfying assignment is $(x_1, x_2, x_3, x_4) = (1, 0, 1, 0)$. 
The corresponding clause-variable incidence graph is depicted with a rectilinear embedding in~\cref{fig:drawing_graph}. 

\begin{figure}[htb]
	\centering
	\includegraphics[page=5]{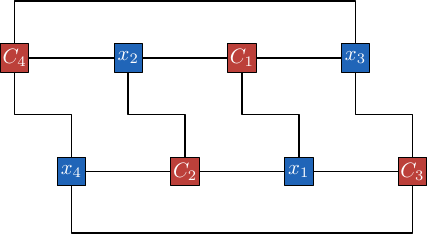}
	\caption{Clause-variable incidence graph with rectilinear embedding.}
	\label{fig:drawing_graph}
\end{figure}

\subsection{Gadgets for office-like polygons with vertices at arbitrary rational coordinates} \label{app:gadgets-hardness-arbitrary}
We explain the gadgets used to prove that it is \NP-hard to decide if a dispersion distance of~$2+\varepsilon$ is realizable for every fixed choice of $\varepsilon \in (0,1)$. 
For completeness, we highlight that larger~$\varepsilon$ are clearly possible by scaling up the entire construction.

Let $0 < \tau < \varepsilon / 2$ and let $\ell := 2 + \varepsilon$. 
In the subsequent figures, we use the following conventions:
\begin{itemize}
	\item The long sides of red rectangles have length $1$, and the short sides have length~$\tau$.
	\item All corridors fully contained within a single gadget are unit squares. They are indicated by dark-gray squares and are exactly the corridors that are not marked with a black arrow. The short side of every corridor marked with a black arrow has length~$1$.
	\item Additional important lengths are specified inside the respective figure. If no information is given for some length, it is either trivially determined by the other lengths, or else should be chosen as short as possible. In particular, this implies that the short sides of most rooms have length~$1$.
\end{itemize}

Subsequently, we introduce several gadgets and discuss guard sets for them. 
For clarity, a guard set is not required to cover rooms or corridors with open ends, though it may do so. 

Similar to the reduction for office-like polygons with vertices at integer coordinates, we introduce a \emph{banning-type subpolygon}; see~\cref{fig:banning_valid_placement}. 
Its main use is to forbid certain guard positions. 
Note that the gadget is essentially a $10$-$\varepsilon$-packing, as discussed in~\cref{10-piece}.
We observe the following:

\begin{observation}\label{obs:banning_arbitrary}
	For the banning-type subpolygon, there exists a unique guard set with a dispersion distance of at least~$\ell$; see \cref{fig:banning_valid_placement}. Therefore, no additional guard can be placed within a distance of~$\ell$ from~$v$.
\end{observation}

\begin{figure}[htb]
	\centering
	\includegraphics[page=6]{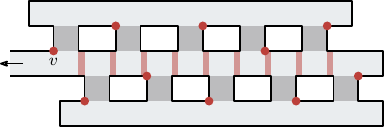}
	\caption{Banning-type subpolygon.}
	\label{fig:banning_valid_placement}
\end{figure}

We now explain the used gadgets.

\subparagraph{Variable gadget.}
To represent variables that can be set to either \texttt{true} or \texttt{false}, we introduce a variable gadget with two $\ell$-realizing guard sets representing the respective states of the variable; see~\cref{fig:2eps_var_gadget}. 
If a variable is set to \texttt{true}, the corresponding guard set will see the corridor labeled with \texttt{true}, but not the two corridors labeled with \texttt{false}. 
Conversely, if the variable is set to \texttt{false}, the corresponding guard set will see both corridors labeled with \texttt{false}, but not the corridor labeled with \texttt{true}. 
Note that the gadget has one \texttt{true}-corridor and two \texttt{false}-corridors because, in \psat, every variable occurs exactly once as positive literal and twice as negative literal in the clauses.

\begin{figure}[htb]
	\centering
	\includegraphics[page=7]{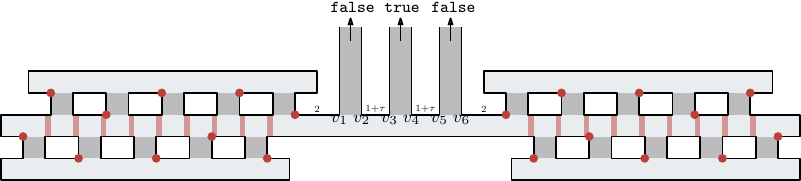}
	\caption{Variable gadget.}
	\label{fig:2eps_var_gadget}
\end{figure}

\begin{lemma}\label{vargadget}
	In the variable gadget, there exist multiple $\ell$-realizing guard sets. 
	One of them covers the corridor labeled with \texttt{true}, but none of the corridors labeled with \texttt{false}. 
	The other one covers both corridors labeled with \texttt{false}, but not the corridor labeled with \texttt{true}. 
	Moreover, no $\ell$-realizing guard set covers the corridor labeled with \texttt{true}, and a corridor labeled with \texttt{false}.
\end{lemma}

\begin{proof}
	First, we show that the two guard sets exist. 
	Let $R$ denote the set of points depicted red in~\cref{fig:2eps_var_gadget}. 
	Clearly, $R \cup \{v_2, v_5\}$ results in an $\ell$-realizing guard set that sees both corridors labeled with \texttt{false}. 
	Moreover, $R \cup \{v_3\}$ is an $\ell$-realizing guard set that sees the corridor labeled with \texttt{true}.
	
	It remains to be shown that no $\ell$-realizing guard set can see a corridor labeled with \texttt{false} and the corridor labeled with \texttt{true}. 
	To this end, note that the variable gadget is constructed by two banning-type subpolygons that share their vertically middle room. 
	From \cref{obs:banning_arbitrary}, it follows that every guard set contains~$R$. 
	Consequently, no guard can be placed at~$v_1$ or~$v_6$. 
	Observe that when placing a guard at~$v_2$, no guard can be placed at~$v_3$ or~$v_4$ because both vertices are too close. 
	The same holds for a guard placed at~$v_5$, which proves the statement.
\end{proof}

\subparagraph{Clause gadget.}
After discussing the variable gadget, we introduce a gadget that represents clauses.  
We construct the clause gadget such that all incoming corridors except for one can be seen from inside the gadget.
Recall that \psat is in \P~for the special case where every clause consists of exactly three literals~\cite{papadimitriou}. 
Therefore, we also introduce a clause gadget that represents clauses consisting of two literals. 
We begin with the $3$-clause gadget, which is depicted in~\cref{fig:2eps_3clause_gadget}.

\begin{figure}[htb]
	\centering
	\includegraphics[page=9]{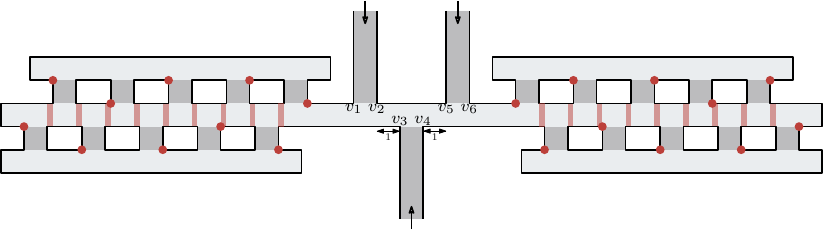}
	\caption{3-literal clause gadget.}
	\label{fig:2eps_3clause_gadget}
\end{figure}

\begin{lemma}
	In the 3-clause gadget, for every pair of incoming corridors, there exists an $\ell$-realizing guard set that sees both of them. 
	Moreover, no $\ell$-realizing guard set can see all three incoming corridors.
\end{lemma}

\begin{proof}
	Consider~\cref{fig:2eps_3clause_gadget}. 
	Let $R$ denote the set of points depicted red. 
	Then, the guard sets $R \cup \{v_2, v_4\}$, $R \cup \{v_3, v_5\}$, and $R \cup \{v_2, v_5\}$ are $\ell$-realizing. 
	Each pair of incoming corridors is seen by one of them.
	
	It remains to be shown that no guard set sees all three incoming corridors. 
	Observe that the two banning-type subpolygons prevent guards from being placed at $v_1$ and $v_6$. 
	The statement follows directly: the incoming corridors are independent and no triple of points from $\{v_2,v_3,v_4,v_5\}$ allows to realize a dispersion distance of~$\ell$.
\end{proof}

Next, we proceed with the $2$-clause gadget.

\begin{figure}[htb]
	\centering
	\includegraphics[page=8]{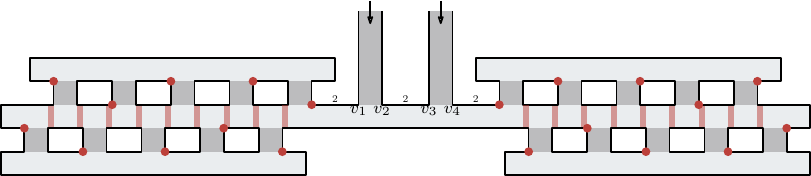}
	\caption{2-literal clause gadget.}
	\label{fig:2eps_2clause_gadget}
\end{figure}

\begin{lemma}
	In the 2-clause gadget, for every incoming corridor, there exists an $\ell$-realizing guard set that sees this corridor. 
	Moreover, no $\ell$-realizing guard set can see both incoming corridors.
\end{lemma}

\begin{proof}
	Consider~\cref{fig:2eps_2clause_gadget}. 
	Let $R$ denote the set of points depicted red. 
	Then, $R \cup \{v_2\}$ and $R \cup \{v_3\}$ are $\ell$-realizing guard sets that each see exactly one incoming corridor.
	
	Observe that the two banning-type subpolygons prevent guards from being placed at~$v_1$ and~$v_4$, and that $\delta(v_2, v_3) < \ell$. 
	As a matter of fact, no $\ell$-realizing guard set sees both incoming corridors. 
\end{proof}

\subparagraph{Bending gadget.}
When connecting the outgoing corridors of variable gadgets with incoming corridors of clause gadgets, bends may be necessary. 
To this end, we introduce a bending gadget that propagates the signal correctly; see~\cref{fig:2eps_bend_gadget}.

\begin{figure}[htb]
	\centering
	\includegraphics[width=0.6\textwidth, page=10]{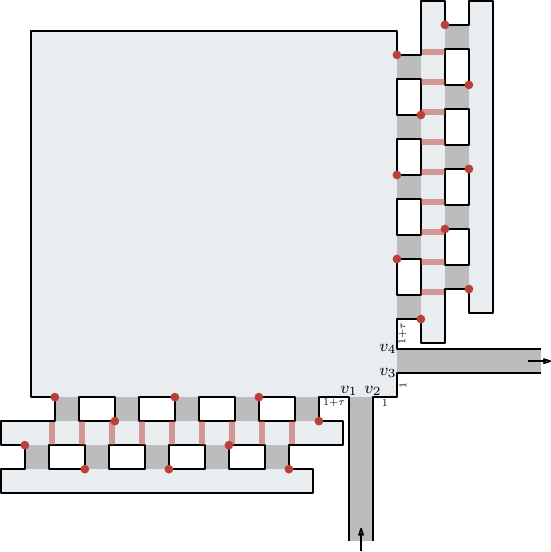}
	\caption{Bending gadget.}
	\label{fig:2eps_bend_gadget}
\end{figure}

\begin{lemma}
	For the bending gadget, there exist multiple $\ell$-realizing guard sets. 
	One of them sees the incoming corridor, but not the outgoing corridor. 
	Another one sees the outgoing corridor, but not the incoming corridor. 
	Moreover, no $\ell$-realizing guard set sees both, the incoming and outgoing corridor.
\end{lemma}

\begin{proof}
	Consider~\cref{fig:2eps_bend_gadget}. 
	Let $R$ denote the set of points depicted red. 
	Clearly, $R \cup \{v_2\}$ and $R \cup \{v_3\}$ are $\ell$-realizing guard sets that either see the incoming or the outgoing corridor.
	
	Moreover, note that there are essentially two banning-type subpolygons that forbid to place guards at~$v_1$ and~$v_4$. 
	Because $\delta(v_2, v_3) < \ell$, no $\ell$-realizing guard set sees both, the incoming and the outgoing corridor.
\end{proof}

\subparagraph{Widening gadget.}
While, in principle, all gadgets needed to encode a given instance of \psat have been introduced, there is a difficulty when connecting these gadgets later:
The bending gadget has too large width compared to the space between incoming and outgoing corridors of the variable and clause gadgets.
To address this, we introduce a widening gadget that, similar to the bending gadget, correctly propagates the signal and thereby creates space between the corridors; see~\cref{fig:widening_valid_placement}. 
This gadget is narrow enough to be connected to the incoming and outgoing corridors of variable and clause gadgets.
As depicted in~\cref{fig:widening_use}, the widening gadget can be used in combination with the variable gadget to create space for the bending gadgets.

\begin{lemma}
	For the widening gadget, there exists an $\ell$-realizing guard set that sees the incoming corridor and another $\ell$-realizing guard set that sees the outgoing corridor. 
	Moreover, no $\ell$-realizing guard set sees both, the incoming and the outgoing corridor.
\end{lemma}

\begin{proof}
	\cref{fig:widening_valid_placement} depicts an $\ell$-realizing guard set that sees the incoming corridor. 
	Due to symmetry, it is easy to obtain an $\ell$-realizing guard set that sees the outgoing corridor. 
	
	It remains to be shown that no $\ell$-realizing guard set can see both, the incoming and the outgoing corridor. 
	To this end, observe that the widening gadget basically consists of a $9$-$\varepsilon$-packing, and the two open corridors. 
	From \cref{9-piece}, we know that the $9$-$\varepsilon$-packing has exactly two $\ell$-realizing guard sets, one of which is depicted in \cref{fig:widening_valid_placement} (excluding the leftmost guard), and the other one is symmetric. 
	Clearly, the guard set in \cref{fig:widening_valid_placement} prevents placing another guard at a sufficient distance to cover the outgoing corridor.
\end{proof}

\begin{figure}[htb]
	\centering
	\includegraphics[page=13]{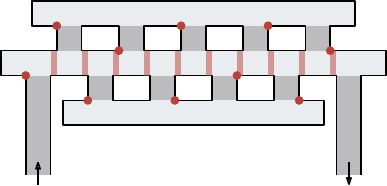}
	\caption{Widening gadget.}
	\label{fig:widening_valid_placement}
\end{figure}

\begin{figure}[htb]
	\centering
	\includegraphics[page=16]{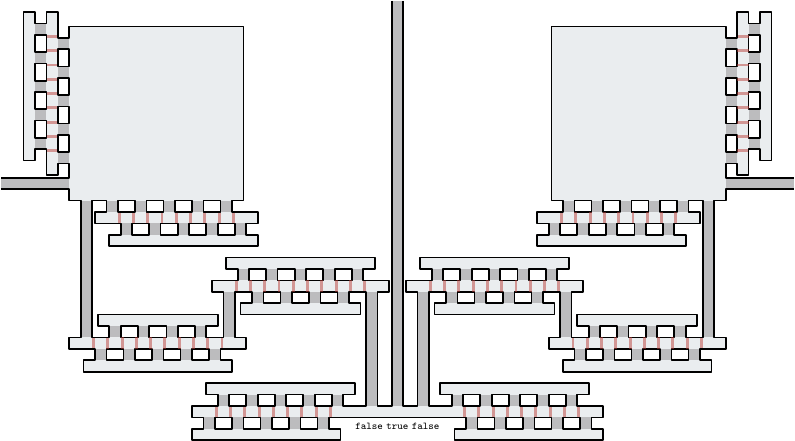}
	\caption{Widening gadgets ensure that there is enough space for the bending gadgets.}
	\label{fig:widening_use}
\end{figure}

\subsection{Completing the construction and proof} \label{hardness_finish}
After introducing the gadgets, it remains to complete the construction.
We describe the procedure for office-like polygons with unrestricted vertex coordinates.
For polyominoes and office-like polygons with vertices at integer coordinates, the construction works similar.

\begin{lemma}\label{hardnes2eps_proof}
	For office-like polygons with holes, it is \NP-hard to decide if a dispersion distance of $2+\varepsilon$ is realizable for any rational $\varepsilon > 0$.
\end{lemma}

\begin{proof} 
	As mentioned earlier, without loss of generality, we assume that~$\varepsilon \in (0,1)$.
	Let $\varphi$ be a formula of \psat.
	We construct an office-like polygon~$\polygon_\varphi$ that can be guarded with a dispersion distance of at least $2+\varepsilon$ if and only if $\varphi$ is satisfiable. 
	We obtain $\polygon_\varphi$ as follows: For the clause-variable incidence graph of~$\varphi$, compute a rectilinear embedding (with polynomially many edge bends). 
	Given the embedding, replace vertices corresponding to variables by variable gadgets, and vertices corresponding to clauses by clause gadgets.
	We choose the correct clause gadget based on the degree of the respective vertex. 
	If a variable occurs in a clause, connect their corresponding gadgets. 
	To this end, follow the edges of the embedded graph and introduce bends where necessary. 
	To create enough space for placing bending gadgets, use widening gadgets before or after the respective variable or clause gadget, as shown in~\cref{fig:widening_use}.
	All remaining (straight) connections between two gadgets are made by extending the incoming and outgoing corridors until they meet. 
	Of course, when a variable occurs as positive literal in a clause, then the outgoing corridor labeled with \texttt{true} should be connected to the respective clause gadget. 
	Vice versa, if the variable occurs as negative literal in the clause, then the variable gadget should be connected by a corridor labeled with \texttt{false} to the respective clause gadget. 
	For an exemplary overview, consider the polygon depicted in~\cref{fig:exemp_scot}.
	
	The construction has polynomial size: 
	For an instance of \psat consisting of~$n$ variables, there are at most $\frac{3n}{2}$ clauses and hence the clause-variable incidence graph has at most~$\frac{5n}{2}$ vertices. 
	Each gadget has constant size. Each vertex is replaced by a clause or variable gadget and a constant number of widening and bending gadgets. 
	To this end, note that four widening gadgets and five bending gadgets always suffice to ensure the outgoing corridors have the correct orientation; see~\cref{fig:exemp_scot}. 
	Additionally, there are only polynomially many edge-bends in the clause-variable incidence graph; thus, we need only polynomially many bending gadgets.
	
	\begin{claim}
		If $\varphi$ is satisfiable, then there exists a guard set that realizes a dispersion distance of $2+\varepsilon$ for $\polygon_\varphi$.
	\end{claim}
	\begin{claimproof}
	Consider an assignment that satisfies $\varphi$. 
	We show how to obtain a guard set for $\polygon_\varphi$ that realizes a dispersion distance of $2+\varepsilon$. 
	For each variable that is set to \texttt{true}, place guards in the respective variable gadget such that they see the corridor labeled with \texttt{true}. 
	For each variable set to \texttt{false}, place guards in the respective variable gadget such that they see both corridors labeled with \texttt{false}. 
	Such guard placements with sufficient dispersion distances exist by~\cref{vargadget}. 
	After placing the guards in each variable gadget, we propagate the signal until it reaches the respective incoming corridor of the clause gadget. 
	More specifically, for each outgoing corridor that is seen from insde the variable gadget, we use the previously explained guard positions of the widening and bending gadgets, ensuring that the incoming corridor to the next clause gadget is guarded from outside. 
	After propagating all signals, each clause must have at least one incoming corridor that is already guarded. 
	Here we use that the given assignment satisfies $\varphi$.
	By construction, we can then always find a guard set for the respective clause gadget that also sees all so-far uncovered incoming corridors. 
	We choose such a guard set for each clause gadget. 
	Finally, for each incoming corridor of clause gadgets that is only seen by a guard placed inside the respective clause gadget, we propagate the signal (in the other direction) through the widening and bending gadgets, until a variable gadget is reached. 
	By construction, every point of $\polygon_\varphi$ is seen by at least one guard and the dispersion distance is at least $2+\varepsilon$. 
	\end{claimproof}
	
	\begin{claim}
		If there exists a guard set that realizes a dispersion distance of $2+\varepsilon$ for $\polygon_\varphi$, then $\varphi$ is satisfiable.
	\end{claim}
	
	\begin{claimproof}
	From a guard set with dispersion distance at least $2+\varepsilon$ for $\polygon_\varphi$, we determine a satisfying assignment for $\varphi$. 
	To this end, we consider all variable gadgets and the guards placed in them. 
	For each variable gadget, if the corridor labeled with \texttt{true} is seen by some guard placed inside the variable gadget, then set the corresponding variable to \texttt{true}. 
	Otherwise, set the variable to \texttt{false}. 
	
	We show that the resulting assignment satisfies $\varphi$. 
	To this end, recall that at least one incoming corridor of each clause gadget must be seen by some guard outside the respective clause gadget. 
	Since widening gadgets and bending gadgets only propagate signals, it follows that there must be a guard placed inside some variable gadget that initializes the signal. 
	As a consequence, the clause is satisfied by our assignment.
	\end{claimproof}
	
	This concludes the proof.
\end{proof}

Interestingly, the \agp turned out to be easy for independent office-like polygons~\cite{cruz}. 
In~particular, for an independent office-like polygon $\polygon = (\mathcal R, \mathcal C)$, a smallest guard set always has size $|\mathcal R| = |\mathcal C| + 1$ if~$\polygon$ is hole-free, and otherwise has size $|\mathcal C|$.

\begin{figure}
	\centering
	\includegraphics[page=17, scale=0.9]{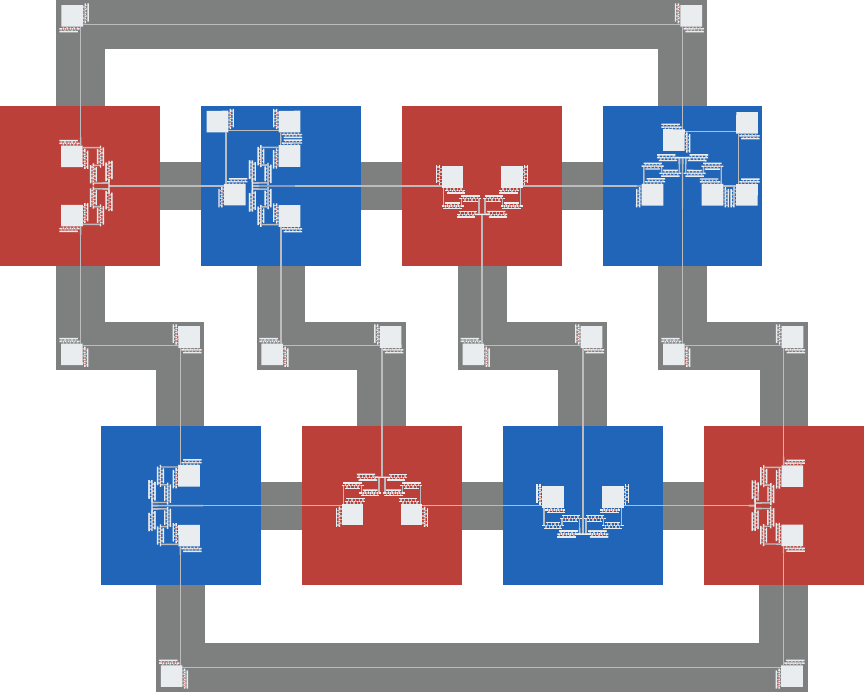}
	\caption{The office-like polygon as a result of the \NP-hardness construction for the Boolean formula $\{x_1 \vee x_2 \vee \overline{x_3}\} \wedge \{\overline{x_1} \vee \overline{x_2} \vee x_4\} \wedge \{\overline{x_1} \vee x_3 \vee \overline{x_4}\} \wedge \{\overline{x_2} \vee \overline{x_3} \vee \overline{x_4}\}$; we refer to~\cref{fig:drawing_graph} for the respective rectilinear embedding of the clause-variable incidence graph.}
	\label{fig:exemp_scot}
\end{figure}
	\section{Omitted details of~\cref{dp_algo}: Dynamic programming}
\label{dp_algo_x}
Computing optimal solutions for the \dagp turns out to be challenging. 
The only non-trivial class of polygons for which a polynomial-time algorithm is known are polyominoes whose dual graph is a tree~\cite{rieck-scheffer-dispersiveAGP}. 
Additionally, the complexity results imply that the existence of a polynomial-time algorithm is unlikely, even for independent office-like polygons.

However, we now give a polynomial-time dynamic programming algorithm that computes the optimal solution in hole-free independent office-like polygons.
The main component of this approach is to solve geometric independent set problems.

\dpalgo*

We start by stating some preliminary assumptions.  
Afterward, in \cref{alg_overview}, we introduce the overall structure of the algorithm and define proper subproblems, which we then solve recursively. 
Furthermore, we explain how to solve each subproblem to obtain the desired solution in~\cref{alg_subproblems}

\subparagraph{Preliminaries.}
Consider a hole-free independent office-like polygon~$\polygon$. 
We assume that~$\polygon$ is given in a way that allows access to the clockwise order of corridors around each room in constant time.
If not provided, this information can easily be obtained; for example by computing the visibility polygon of room corners (i.e., of convex vertices), and then keeping track of which vertices lie on the boundary.

We define $G(\polygon) = (V,E)$ as the graph representing~$\polygon$ as follows: 
For each room $R_v$ of $\polygon$ there is a node $v \in V$. 
Furthermore, there is an edge $e=\{v,w\} \in E$ if the rooms $R_v$ and~$R_w$ are connected by a corridor $C_e$, see \cref{fig:overview_dp}. 
Note that, since $\polygon$ is hole-free, $G(\polygon)$ is a tree.

Furthermore, we assume that the distance between each pair of vertices in~$\polygon$ is given. 
This information can be computed in linear time~\cite{bae}.

\subsection{High-level overview} \label{alg_overview}
Consider a hole-free independent office-like polygon~$\polygon$ with~$n$ vertices. 
We solve the decision problem whether a guard set with a dispersion distance of at least $\ell$ exists. 
Afterward, we do a binary search over the $\mathcal{O}(n^2)$ many possible dispersion distances (implied by the pairwise vertex distances in $\polygon$) to solve the corresponding maximization problem, i.e., finding a guard set realizing a maximum dispersion distance by using our algorithm for the decision problem as a subroutine.

The algorithm works as follows: 
Consider an arbitrary in-arborescence $G'$ for $G(\polygon)$; for example, consider \cref{fig:overview_dp}. 
We define a subproblem for each node in $G'$ and solve it once all predecessor nodes have been marked as processed.

For the runtime, to solve the decision problem, we have to solve a subproblem at every node of $G'$. 
Moreover, we add a logarithmic factor to the runtime for the binary search on the possible dispersion distances to obtain an optimal solution for the maximization problem. 
Hence, we result with a total runtime in $\mathcal{O}(t_n\cdot n \log n)$, where $t_n$ denotes the time needed to solve the subproblem.
To this end, in \cref{alg_subproblems}, we discuss an algorithm for the subproblem, which has an amortized runtime in $\mathcal{O}(n^4 \log n)$ across all subproblem calls. 
As a consequence, the total runtime of our algorithm is in $\mathcal{O}(n^5 \log^2 n)$.

Before we explain the subproblem in detail, we introduce important concepts.

\begin{figure}[htb]
    \centering
    \includegraphics[page=2]{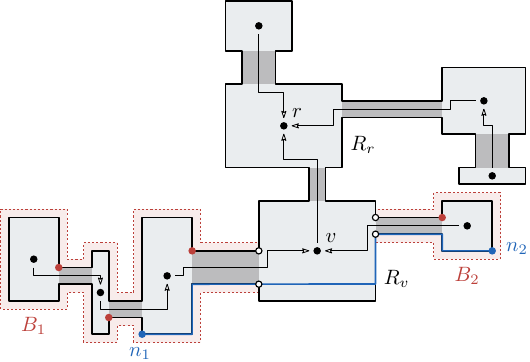}
    \caption{Hole-free independent office-like polygon~$\polygon$ with in-arborescence $G'$ for $G(\polygon)$ and branches $B_1,B_2$ that are rooted at $v$. Note that this is a copy of~\cref{fig:overview_dp_outline}.}
    \label{fig:overview_dp}
\end{figure}

\subparagraph{Branches and configurations.}
Consider a directed edge $e=(u,v)$ in $G'$ and the maximal subtree~$T$ of~$G'$ rooted at~$u$. 
The \emph{branch}~$B$ for~$e$ consists of~$C_e$ and all rooms and corridors corresponding to nodes and edges of~$T$. 
We refer to the two vertices that are shared by~$C_e$ and~$R_v$ as \emph{gate} vertices for~$B$. 
For convenience, we define the branches rooted at~$v$ as those for edges of the form $(\cdot,v)$. 
In \cref{fig:overview_dp}, the rooms and corridors outlined in red correspond to the branches rooted at~$v$, and gate vertices are marked white. 

A guard set for~$B$ contains guards exclusively on vertices within~$B$, ensuring coverage of every room and corridor in~$B$. 
Note that some guard sets for~$B$ also cover~$R_v$ (by placing a guard at a gate vertex), while others do not.
A set of \emph{configurations}~$C$ for~$B$ is a collection of guard sets for~$B$ with a dispersion distance of at least~$\ell$, and with the following property: 
If existent, consider an arbitrary guard set~$\guardset$ for~$\polygon$ with a dispersion distance of at least~$\ell$. 
Let $\guardset_B \subseteq \guardset$ denote the guards that are placed at vertices within~$B$. 
Then, there always exists a configuration $c \in C$ such that $(\guardset \setminus \guardset_B) \cup c$ forms a guard set for~$\polygon$ with a dispersion distance of at least~$\ell$.

\subparagraph{Subproblems.} 
When processing the root~$r$ of $G'$, we solve the entire decision problem. 
For every other node~$v$ in~$G'$, we define a subproblem as follows. 
Let~$w$ denote the unique node such that there exists an edge from~$v$ to~$w$ in~$G'$. 
Find a set of configurations for the branch of $(v,w)$. 
Note that, by the time we process~$v$, all predecessor nodes have already been processed, and we therefore know a set of configurations for each branch rooted at~$v$. 
We discuss later that it suffices to consider configuration sets of constant size. 
Moreover, if at some point a computed configuration set is empty, this implies that no guard set with a dispersion distance of at least~$\ell$ exists for~$\polygon$, not even within this branch.

\subsection{Details on subproblems} \label{alg_subproblems}
We now discuss how to solve the previously stated subproblem for a node~$v$ in~$G'$.
More specifically, if~$v$ is the root of~$G'$, we solve the entire decision problem.
Otherwise, we compute a configuration set for the branch~$B$ belonging to the (unique) edge $e = (v,w)$.
We will discuss the case where~$v$ is the root of~$G'$ later; for now, we consider the alternative.
Using that~$\polygon$ is independent, a crucial observation is the following.

\gateVerticesShortestPath*

Because $\polygon$ is independent, at least one guard needs to be placed in every corridor.
Furthermore, it suffices to place 
\begin{itemize}
    \item at most two guards in every corridor, and if so, in different rooms.
    \item at most one guard at corners of each room.
\end{itemize}
This holds because otherwise there would always be two guards with identical visibility polygons.
Building on these observations, we define a configuration set of constant size for~$B$.

\subparagraph{Configuration set.}
Clearly, there is only a constant number of combinations to place guards at the corners of~$R_v$ and in~$C_e$; some of these can even be ignored due to the above observations, or because there are guards in distance~$\ell$.
In \cref{configurations}, sensible guard placements are depicted in red.
This list is complete, up to symmetric placements and alternative choices for placing a guard at a corner of~$R_v$.
We refer to the set of these placements as~$\mathcal{X}$.

For each placement $X \in \mathcal{X}$, we aim to construct a configuration, i.e., a special guard set for $B$.
The configuration is basically constructed as follows: choose one configuration for every branch rooted at~$v$ such that their union, together with $X$, is a guard set for~$B$. 
If~$X$ does not contain a guard in~$R_v$, then we need to ensure that at least one of the chosen configurations (for branches rooted at~$v$) contains a guard at a gate vertex of its respective branch. 
However, there are many possibilities to choose from, and some may be better than others.
In particular, in our constructed configuration, no two guards should have a distance that is smaller than $\ell$.
If this is not possible, then we do not need to construct a configuration for~$X$.
If there are even options remaining, then we choose configurations so that the minimum distance of a guard placed in any of the chosen configurations to the respective blue vertex in~$C_e$, which is a gate vertex of~$B$, is maximized; see \cref{configurations}.
(For the case where two guards are placed on diagonal vertices of $C_e$, this is clearly not necessary.)
Note that, if the distance to the blue vertex (which is a gate vertex of~$B$) is maximized, so is the distance to the other gate vertex of~$B$.
By \cref{gate_dist_main}, this is best possible when using~$X$.

\begin{figure}[htb]
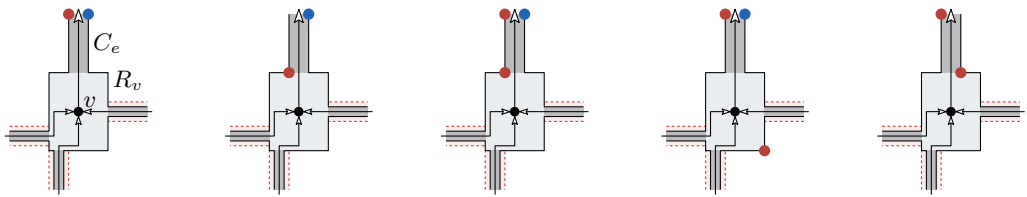

    \centering
    \begin{subfigure}[b]{0.18\textwidth}
         \centering
         \includegraphics[page = 1]{figures/configurations.pdf}
         \label{conf1-app}
     \end{subfigure}
     \hfill
     \begin{subfigure}[b]{0.18\textwidth}
         \centering
         \includegraphics[page=2]{figures/configurations.pdf}
     \end{subfigure}
     \hfill
     \begin{subfigure}[b]{0.18\textwidth}
         \centering
         \includegraphics[page = 3]{figures/configurations.pdf}
     \end{subfigure} 
     \hfill
     \vspace{0.5cm}
     \begin{subfigure}[b]{0.18\textwidth}
         \centering
         \includegraphics[page = 4]{figures/configurations.pdf}
         \label{corner_place-app}
     \end{subfigure}
     \hfill
     \begin{subfigure}[b]{0.18\textwidth}
         \centering
         \includegraphics[page = 5]{figures/configurations.pdf}
         \label{no_blue-app}
     \end{subfigure}
    \caption{Schematic illustration of sensible guard placements at corners of $R_v$ and in $C_e$ (depicted red).
        When fixing one of the options, it suffices to find a guard set for $e$'s branch that has a dispersion distance of at least~$\ell$ and, among those, maximizes the distance to the respective blue vertex.
        Note~that this figure is the same as~\cref{configurations_main}, shown again for clarity and convenience.}
    \label{configurations}
\end{figure}

It remains to discuss how to compute such a configuration for $X$ in polynomial time. 
 
\subparagraph{Solving subproblems.}
First, for every branch rooted at~$v$, we discard all configurations in which a guard has a distance of less than~$\ell$ to a guard in~$X$. 
This can be done in linear time: 
There are at most two guards in~$X$, and at most $\mathcal{O}(n)$ many guards in configurations to consider, as there are a constant number of configurations for each branch rooted at~$v$. 
For every pair of guards, one from~$X$ and one from the respective configuration, we check if the distance is less than~$\ell$. 
If so, we discard the entire configuration.

Second, we perform a binary search over the $\mathcal{O}(n^2)$ pairwise vertex distances in~$\polygon$, aiming to find the choice of configurations for branches rooted at~$v$ that maximizes the distance to the blue vertex, as discussed above. 
For each distance~$\ell'$ considered during the current step of the binary search, we discard all configurations of branches rooted at~$v$ where a guard has a distance of less than~$\ell'$ to the blue vertex. 
This can be done in linear time, similar to the procedure above. 
Moreover, as mentioned before, depending on whether a guard in~$X$ sees~$R_v$, we may need to ensure that one of the chosen configurations contains a guard at a gate vertex of its branch. 
This can be done by selecting a fixed configuration that contains a guard at a gate vertex of its branch~$B'$, then discarding all other configurations belonging to~$B'$. 
Of course, we must check every possible such configuration, separately.

The crucial part is now to select configurations (from the remaining ones), one for each branch rooted at $v$, such that no two guards in different configurations have a distance of less than~$\ell$, or decide that no such configurations exist. 
In this case, we no longer need to consider~$X$.
Note that this problem corresponds to finding an independent set of size~$k$, where~$k$ denotes the number of branches rooted at~$v$, for the following graph:
Each (non-discarded) configuration of a branch rooted at~$v$ is a node. 
Two nodes are connected by an edge if either both corresponding configurations $c_1, c_2$ belong to the same branch, or if there exists a guard $g_1 \in c_1$ and a guard $g_2 \in c_2$ such that $\delta(g_1,g_2) < \ell$.
Below, we state a simple algorithm that solves the independent set problem in polynomial time.\\

We are left to consider the case where~$v$ is the root of~$G'$. 
This case is nearly identical to the others. 
The only difference is the set of positions~$\mathcal X$. 
Here, we only have to consider placements where either none or exactly one guard is at a corner of~$R_v$, since there is no corridor remaining that we need to take care of. 
In the first case, we again need to ensure that we choose at least one configuration for a branch rooted at~$v$ that places a guard at a gate vertex.

We discuss the overall runtime for the subproblem associated with node~$v$ in~$G'$.
There are constantly many placements in~$\mathcal{X}$ to consider.
For each placement, we have to solve the independent set problem at most $\mathcal{O}(k \cdot \log n)$ times, where~$k$ is the number of corridors incident to~$R_v$, and the remainder represents the binary search.
Note that the~$k$ comes from fixing configurations that ensure coverage of~$R_v$.
On average, among all subproblem calls for nodes in~$G'$, $k$ clearly is constant.
Each independent set problem can be solved with a runtime in~$\mathcal O(n^4)$, which we discuss below.
Therefore, the problem associated with node~$v$ can be solved with the already mentioned amortized runtime in~$\mathcal{O}(n^4 \log n)$.

\subsection{Computing independent sets} \label{alg_independent_sets}
\newcommand{\BVe}[0]{BV}
\newcommand{\BHo}[0]{BH}
We discuss a simple algorithm to solve the above stated independent set problem, making use of the given geometric representation.
We begin by introducing important terminology, followed by discussing a simple subroutine performing a greedy selection of configurations.
Afterward, we explain the full procedure.

The distance between two configurations $c_1, c_2$ (for different branches) is the minimum pairwise distance between a guard from~$c_1$ and a guard from~$c_2$.
Similarly, the distance between a configuration~$c$ and a point~$p$ is the minimum distance between any guard in~$c$ and $p$.
Note that we can compute any of these distances in constant time for configurations belonging to branches rooted at~$v$, since we already know the smallest distance from a guard in such a configuration~$c$ to its respective gate vertex:
the information was obtained at the moment when~$c$ was computed.

Among the branches rooted at~$v$, $\BVe$ consists of those that extend vertically, i.e., either to the top, or to the bottom from~$R_v$ (meaning their corridors incident to~$R_v$ extend in this direction).
Similarly, $\BHo$ contains all branches that extend horizontally.

Recall that, at this point, we have a configuration set~$C_i$ of constant size for each branch~$B_i$ rooted at~$v$.
We denote by~$C_{\BVe}$ (or~$C_{\BHo}$) the collection of all configuration sets for branches in~$\BVe$ (or~$\BHo$).
With a slight abuse of notation, deleting a configuration~$c_i$ for branch~$B_i$ from~$C_{\BVe}$ actually means deleting it from~$C_i \in C_{\BVe}$.

For~$B_1, B_2 \in \BVe$, we say that~$B_1$ is to the left of~$B_2$ if there exists a vertical line such that the two gate vertices of~$B_1$ lie to its left, while the two gate vertices of~$B_2$ lie to its right.

For $B_i \in \BVe$, consider its configurations $C_i$.
Let $v, v'$ denote the left and right gate vertex of~$B_i$.
For~$c_i \in C_i$, $\delta_{c_i}(v)$ is the shortest distance from a guard in~$c_i$ to~$v$; $\delta_{c_i}(v')$ is defined similarly.
Consider two configurations $c_i, c_j \in C_i$. 
Without loss of generality, we assume that either $\delta_{c_i}(v) > \delta_{c_j}(v)$ or $\delta_{c_i}(v') > \delta_{c_j}(v')$, since otherwise~$c_j$ would always be preferable, allowing us to discard~$c_i$ immediately. 
We say that~$c_i$ is \emph{smaller} than~$c_j$ if $\delta_{c_i}(v') > \delta_{c_j}(v')$. 

Next, we introduce a simple subroutine that greedily selects configurations for branches extending to opposite sides from~$R_v$.

\subparagraph{Greedy selection.}
We assume either $(\BVe, C)$ or $(\BHo, C)$ as input for the greedy routine, where~$C$ contains configuration sets for the branches in $\BVe$ or $\BHo$, respectively.
We now explain the routine for~$(\BVe, C)$; note that the other case is identical, except that~$\polygon$ is effectively rotated~\ang{90} clockwise.

Consider the branches $B_1, \dots, B_k$ in $\BVe$, sorted from left to right.
Let $C_i \in C$ denote the set of configurations for~$B_i$. 
We select a configuration for each branch in~$\BVe$ as follows: 
Consider~$B_i$ for increasing~$i$, and select the smallest configuration in~$C_i$ that has a distance of at least~$\ell$ to all previously selected configurations.

If no configuration can be selected at any point, return \texttt{false}.
After selecting a configuration for each branch in~$\BVe$, return \texttt{true}.

Next, we analyze the runtime of the greedy routine.

\begin{claim} \label{cl:greedy_runtime}
    The greedy selection can be implemented with a runtime in~$\mathcal{O}(n)$.
\end{claim}

\begin{claimproof}
    In the iteration where~$B_i$ is considered, the corresponding configuration can be selected in constant time:
    Let $B^t$ (or $B^b$) denote the rightmost branch to the left of~$B_i$ among those that extend to the top (or bottom) from~$R_v$.
    Observe that a shortest geodesic $L_1$-path from a vertex in~$B_i$ to any vertex in the branches of~$\BVe$ to the left of~$B_i$ always passes through the right gate vertex of either~$B^t$ or~$B^b$.
    Therefore, it suffices to keep track of which previously selected configurations have a minimum distance to these gate vertices.
    Then, we can select the smallest configuration for~$B_i$ that has a distance of at least~$\ell$ to the configurations stored with~$B^t$ and~$B^b$; this can be done in constant time, as there are only a constant number of configurations available for~$B_i$.

    Note that, in iteration $i+1$, $B_i$ will replace either $B^t$ or $B^b$ as reference for the next iteration.
    Updating the stored configurations is straightforward, as the configuration with the minimum distance to the respective gate vertex will either be the one from the previous iteration, or the one added in this iteration.

    Since there are~$\mathcal{O}(n)$ configurations to consider, the linear runtime follows immediately.\claimqedhere
\end{claimproof}

By construction, if a feasible solution exists, the greedy algorithm will find it.
It is easy to see that the selected configurations have the property that no smaller configuration can be chosen for any of the branches while still obtaining a feasible solution.
This implies that the greedily selected configurations maximize the distance to the top-right and bottom-right corners of~$R_v$.

We proceed by presenting the algorithm in full detail.

\subparagraph{Algorithm for the independent set problem.}
Let~$D$ denote the set of all pairwise geodesic $L_1$-distances between vertices of~$\polygon$.
For every $p \in D$, we apply the following routine:
\begin{enumerate}
    \item Compute $C'_{\BVe}$ (or $C'_{\BHo}$) by removing all configurations from $C_{\BVe}$ (or $C_{\BHo}$) that have a distance of less than~$p$ (or $\ell - p$) to the bottom-left corner of~$R_v$.
    \item \label{enum:greedy} Apply the greedy selection for~$(\BVe, C_{\BVe}')$ and~$(\BHo, C_{\BHo}')$. 
        If any of the greedy selections returns \texttt{false}, then there is no solution for~$p$.
        Hence, continue with the next probe.
    \item If every pair of greedily selected configurations has a distance of at least~$\ell$, then we have found a solution and terminate the algorithm.
            Otherwise, let~$c_i, c_j$ denote two selected configurations for branches~$B_i, B_j$ with a distance smaller than~$\ell$.
            \begin{enumerate}
                \item If $B_i$ extends to the bottom from $R_v$, discard $c_j$ from $C_{\BHo}'$ and return to step~\ref{enum:greedy}.
                \item If $B_i$ extends to the top from $R_v$
                \begin{enumerate}
                    \item and $B_j$ extends to the right from~$R_v$, then there is no solution for~$p$; proceed with the next probe.
                    \item and $B_j$ extends to the left from~$R_v$, then discard $c_i$ from $C_{\BVe}'$ and return to step~\ref{enum:greedy}.
                \end{enumerate}
            \end{enumerate}
            
\end{enumerate}
If no feasible solution was reported after the last probe for~$p$, then terminate with \texttt{false}.

It remains to discuss the correctness and runtime.

\begin{claim}
    The above algorithm can be implemented to solve the independent set problem with a runtime in~$\mathcal{O}(n^4)$.
\end{claim}

\begin{claimproof}
    For the correctness: Assume that a feasible solution exists; otherwise, there is nothing to prove.
    Considering the optimal solution, there must exist a~$p \in D$ such that all chosen configurations for~$\BVe$ (or~$\BHo$) have a distance of at least~$p$ (or $\ell - p$) to the bottom-left corner of~$R_v$.
    We show that the algorithm will find a feasible solution when probing for exactly this value of~$p$:

    Clearly, the computed $C'_{\BVe}$ and $C'_{\BHo}$ in step~$1$ preserve the possibility of finding a feasible solution (for~$p$); i.e., the chosen configurations of a feasible solution for all branches are still among $C'_{\BVe}$ and $C'_{\BHo}$.
    Assume that $C'_{\BVe}$ and $C'_{\BHo}$ still contain the necessary configurations to obtain a feasible solution before the algorithm invokes the greedy selection in step~$2$ for the $i$-th time.
    We show that, unless the algorithm finds a feasible solution when entering step~$3$ for the next time, this condition holds before the algorithm invokes the greedy selection for the $(i+1)$-th time.

    Clearly, both subsequent greedy selections will return \texttt{true}, because, by assumption, there is a feasible solution among $C'_{\BVe}$ and $C'_{\BHo}$.
    Note that no two greedily selected configurations, belonging to branches extending in opposite or identical directions, can have a distance smaller than~$\ell$.
    Moreover, we leverage the observation mentioned above:
    For $B_i \in \BVe$, the greedily chosen configuration~$c$ maximizes the distance to the top-right and bottom-right corner of $R_v$, among all sensible configurations for~$B_i$.
    In particular, any available configurations for~$B_i$ that has a larger distance to these corners than~$c$ can never be part of a feasible solution.
    Similarly, the configurations chosen greedily for branches in~$\BHo$ maximize the distance to the top-right and top-left corner of~$R_v$.
   
    These properties are crucial for establishing the correctness of the decisions made in step~$3$ of the algorithm.
    However, if the the algorithm finds a feasible solution in step~$3$, then there is nothing to show.
    Otherwise, there exist configurations~$c_i$ and~$c_j$ for branches~$B_i$ and~$B_j$ with a distance smaller than~$\ell$.
    We distinguish all possible combinations of the directions in which~$B_i$ and~$B_j$ could extend from~$R_v$.

    \begin{enumerate}
        \item \textbf{$B_i$ extends to the top, and $B_j$ extends to the right:} 
            Note that a shortest path between guards in $c_i$ and $c_j$ visits the top-right corner of~$R_v$.
            But $c_i$ and $c_j$ maximize the distance to the top-right corner of $R_v$ (among the sensible possibilities for $B_i, B_j$).
            Hence, no solution can exist.
        \item \textbf{$B_i$ extends to the top, and $B_j$ extends to the left:}
            A shortest path between guards in $c_i$ and $c_j$ visits the top-left corner of~$R_v$.
            Since $c_j$ maximizes the distance to the top-left corner, no valid solution can contain $c_i$.
        \item \textbf{$B_i$ extends to the bottom, and $B_j$ extends to the right:}
            A shortest path between guards in $c_i$ and $c_j$ visits the bottom-right corner of~$R_v$.
            Since $c_i$ maximizes the distance to the bottom-right corner, no valid solution can contain $c_j$.
        \item \textbf{$B_i$ extends to the bottom, and $B_j$ extends to the left:}
            This case will not happen due to the construction of $C'_{\BVe}$ and $C'_{\BHo}$ in step~$1$ of the algorithm.
    \end{enumerate}

    Therefore, each decision made by the algorithm in step~$3$ either yields a valid certificate that no solution exists, or discards a configuration that can never be part of a feasible solution.
    Since the algorithm clearly will terminate, this proves the correctness.

    For the runtime: There are~$\mathcal{O}(n^2)$ possible values to probe for.
    Computing~$C_{\BVe}'$ and~$C_{\BHo}'$ in step~$1$ for every probe is straightforward in linear time each.
    During each probe, there are~$\mathcal{O}(n)$ iterations (consisting of repeated applications of steps~$2$ and~$3$), as in each iteration, either a configuration is deleted, or the probe terminates.

    Each iteration has a runtime in~$\mathcal{O}(n)$: two greedy selections are performed, and it is determined whether there exists a pair of configurations with a distance smaller than~$\ell$ among the greedily chosen ones.
    The greedy selections, as discussed by \cref{cl:greedy_runtime}, have a runtime in~$\mathcal{O}(n)$.
    Finding two configurations $c_i$ and $c_j$ for branches $B_i$ and $B_j$ with a distance smaller than~$\ell$, if they exist, is straightforward by using the observation that shortest paths between guards in~$c_i$ and~$c_j$ must pass through a corner of a room:
    For example, during the greedy selection, keep track of which two selected configurations have minimum distance to the corners of~$R_v$, respectively.
    Then, finding a pair $c_i, c_j$ with a distance smaller than~$\ell$ amounts to computing the distance between the two configurations stored at each corner of~$R_v$.
    This can clearly be done in constant time.
\end{claimproof}

This finishes the explanation of our algorithm.
	\section{Omitted details of~\cref{satsolver}}
\label{app:visibility}

Computing the $r$-visibility polygon $\text{Vis}(q)$ in an orthogonal polygon is essential for implementing any of the models described in \cref{satsolver}.
Since no existing implementation was available, we introduce a simple $\mathcal{O}(n \log n)$ algorithm, where the dominant cost is sorting.
We describe the algorithm in a slightly more general setting, where $\text{Vis}(q)$ is computed given~$n$ axis-parallel line segments in the plane.

The basic idea is to partition $\text{Vis}(q)$ into four quadrants $Q_1, \dots, Q_4$ using a vertical and horizontal line through $q$ (\cref{fig:rvis_quadrants}).
Each quadrant is processed similarly; we describe~$Q_1$. 
It is based on the observation that the quadrant~$Q_1$ consists of concatenated rectangles with decreasing height from left to right (\cref{concat_q1}).
To construct $Q_1$, we identify a \emph{canonical} point set $C$ containing the top-left corners of these rectangles and the bottom-right corner of the rightmost rectangle.

To find the canonical point set, initialize $C$ as an empty set and process each segment $s$ with at least one endpoint in $Q_1$:
If $s$ is fully contained in $Q_1$ and horizontal (or vertical), add its left (or bottom) endpoint to $C$.
If $s$ intersects the vertical or horizontal line through~$q$, add the intersection to $C$.

If no point in $C$ lies on the horizontal (or vertical) line through $q$, add $(\infty, q.y)$ (or $(q.x, \infty)$).
Then, filter $C$ to retain only Pareto-optimal points:
discard any $(x_1, y_1)$ dominated by $(x_2, y_2)$ with $x_2 \leq x_1$ and $y_2 \leq y_1$.

To construct $Q_1$, sort $C$ by increasing $x$-coordinate.
Traverse $C$, inserting $(x_2, y_1)$ between consecutive points $(x_1, y_1), (x_2, y_2)$.
Prepend $q$ to complete $Q_1$. 

\begin{figure}[htb]
    \centering
     \begin{subfigure}[b]{0.45\textwidth}
         \centering
         \includegraphics[page=2]{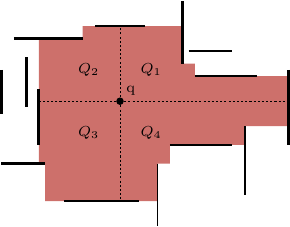}
         \subcaption{}
         \label{fig:rvis_quadrants}
     \end{subfigure}
     \hfill
     \begin{subfigure}[b]{0.45\textwidth}
         \centering
         \includegraphics[page = 5]{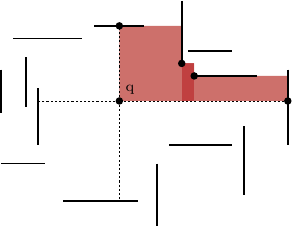}
         \subcaption{}
         \label{concat_q1}
     \end{subfigure}
    \caption{Exemplary illustrations of quadrants and a canonical point set.}
    \label{algo_vis_poly}
\end{figure}
	\section{Details on the ratio of guard sets}
\label{app:small-guard-sets}

We give details on the ratio between the sizes of guard sets in optimal solutions to the classic \textsc{Art Gallery Problem}, and the \textsc{Dispersive AGP}.

\begin{lemma}\label{ratio2scots}
	There exists a family of office-like polygons in which the ratio between the size of a smallest guard set with maximum dispersion and a smallest guard set in general approaches~$2$.
\end{lemma}

\begin{proof}
	For arbitrary $k \geq 2$, we construct an office-like polygon as follows; we refer to~\cref{fig:ratio2construction} for the high-level idea.
	\begin{enumerate}
		\item We place rooms $R_1, \dots, R_k$ with height~$1$ and width~$2k^2 + 4k + 1$. 
            The rooms are stacked vertically from top to bottom in the given order, with their left boundaries aligned along the vertical line $x = 0$.
            Between~$R_1$ and~$R_2$, there are two vertical units of space. The same holds for the rooms~$R_{k-1}$ and~$R_k$. 
            For every other pair of consecutive rooms~$R_i$ and~$R_{i+1}$, there is one vertical unit of space.
		\item For $1 \leq i \leq k-1$, we connect every pair of consecutive rooms $R_i, R_{i+1}$ by corridors $C_1, \dots, C_k$. 
            Every corridor has width~$1$; the height is already determined by the vertical space between the rooms. 
            We place~$C_j$ such that the left boundary is on the vertical line $x= 1 + (j-1) (2k+2)$.
	\end{enumerate}
	Other important lengths implied by this construction are indicated in \cref{fig:ratio2construction}. 

    We discuss possible guard sets for the constructed office-like polygon. 
    First, note that the construction contains several vertical strips $S_1, \dots, S_k$ of consecutive corridors such that, if a guard sees one corridor in~$S_i$, then the guard also sees all other corridors in~$S_i$. 
    Hence, by placing one guard at an arbitrary vertex shared by $S_i$ and $R_i$ for each $1 \leq i \leq k$, we obtain a guard set of size~$k$; see the red guards in \cref{fig:ratio2construction}. 
    Clearly, this is a smallest guard set.
	
	Next, we discuss that every guard set with maximum dispersion has a size of at least~$2k - 2$. 
    To this end, note that the blue vertices in \cref{fig:ratio2construction} form a guard set of this size with a dispersion distance of~$4k + 1$.
    The guard set is constructed as follows:
    Place a guard at a vertex of each strip from left to right, alternating between the top-left and bottom-left vertices of each strip. 
    These guards ensure that every corridor is covered, and thereby the rooms~$R_1, R_k$. 
    To cover the rooms $R_2, \dots, R_{k-1}$, place a guard at the respective top-right~corner.
	
	We show that no guard set containing a guard at a vertex shared by a corridor and one of the rooms $R_2, \dots, R_{k-1}$ can achieve a dispersion distance of~$4k + 1$. 
    It is clear then that no guard set with maximum dispersion can have a size that is smaller than~$2k-2$.

	Assume there exists a guard set with maximum dispersion that places a guard~$g$ at a vertex inside room~$R_i$ for some $2 \leq i \leq k{-}1$, and simultaneously in strip~$S_j$ for $1 \leq j \leq k$. 
    Let~$S'$ denote a neighbor in the ordering of the strips from left to right. 
    Observe that the largest distance from~$g$ to a vertex in~$S'$ is at most~$4k$; for an extreme example, consider the two green vertices in \cref{fig:ratio2construction}. 
    Since at least one guard must be placed in every strip, a dispersion distance of~$4k + 1$ is clearly not achievable.
	
	In summary, there always exists a guard set of size~$k$, but no guard set with maximum dispersion can have a size smaller than $2k - 2$.
    As a matter of fact, the ratio approaches~$2$ as~$k$ increases.
\end{proof}

\begin{sidewaysfigure}
	\centering
	\includegraphics[page=1]{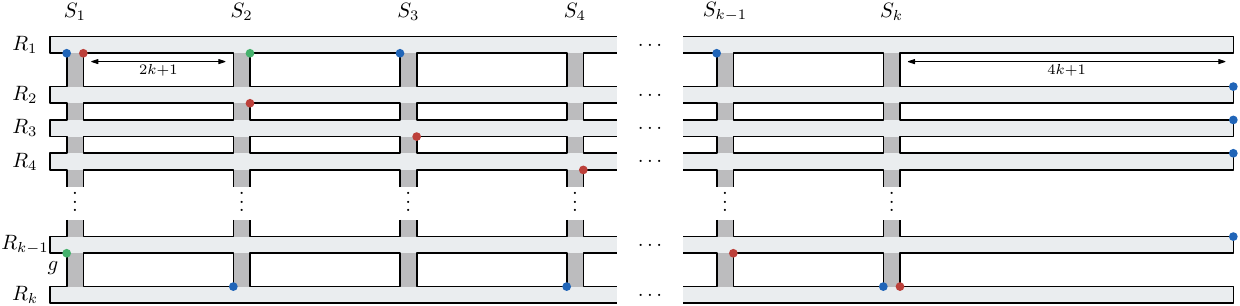}
	\caption{A family of office-like polygons demonstrating that guard sets with maximum dispersion can be nearly twice as large as smallest guard sets.}
	\label{fig:ratio2construction}
\end{sidewaysfigure}

\end{document}